\documentclass[12pt]{article}

\usepackage{geometry} 

\usepackage{bbm}
\usepackage{amsmath}
\usepackage{amsthm}
\usepackage{amssymb}
\usepackage{mathtools}
\usepackage[utf8]{inputenc}
\usepackage{natbib}
\RequirePackage[colorlinks,citecolor=blue,urlcolor=blue]{hyperref}
\usepackage{wrapfig}
\usepackage{graphicx}
\usepackage{caption}
\usepackage{subcaption}
\graphicspath{ {./graphs/} }
\usepackage{multirow}

\theoremstyle{plain}
\newtheorem{theorem}{Theorem}

\newtheorem{lemma}[theorem]{Lemma}
\newtheorem{assumption}[theorem]{Assumption}

\newtheorem{remark}[theorem]{Remark}

\usepackage{algorithm}
\usepackage{algpseudocode}
\algnewcommand{\IIf}[1]{\State\algorithmicif\ #1\ \algorithmicthen}
\algnewcommand{\EndIIf}{\unskip\ \algorithmicend\ \algorithmicif}
\algnewcommand\algorithmicinitialize{\textbf{Initialize:}}
\algnewcommand\algorithmicstepa{\textbf{Step 1:}}
\algnewcommand\algorithmicstepb{\textbf{Step 2:}}
\algnewcommand\algorithmicstepc{\textbf{Step 3:}}
\algnewcommand\algorithmicstepd{\textbf{Step 4:}}
\algnewcommand\algorithmicstepe{\textbf{Step 5:}}
\algnewcommand\algorithmicstepf{\textbf{Step 6:}}
\algnewcommand\algorithmicoutput{\textbf{Output:}}
\algnewcommand\Initialize{\item[\algorithmicinitialize]}
\algnewcommand\Stepa{\item[\algorithmicstepa]}
\algnewcommand\Stepb{\item[\algorithmicstepb]}
\algnewcommand\Stepc{\item[\algorithmicstepc]}
\algnewcommand\Stepd{\item[\algorithmicstepd]}
\algnewcommand\Stepe{\item[\algorithmicstepe]}
\algnewcommand\Stepf{\item[\algorithmicstepf]}
\algnewcommand\Output{\item[\algorithmicoutput]}

\usepackage{graphicx, color}
\graphicspath{{fig/}}

\usepackage{algorithm, algpseudocode} 
\usepackage{mathrsfs} 
\DeclareMathOperator*{\argmax}{arg\,max}

\newcommand{\blind}{0}

\begin{document}

\def\spacingset#1{\renewcommand{\baselinestretch}%
{#1}\small\normalsize} \spacingset{1}


\if0\blind
{
  \title{\bf MaxTDA: Robust Statistical Inference for Maximal Persistence in Topological Data Analysis}
  \author{Sixtus Dakurah\thanks{
    The authors gratefully acknowledge support from NSF grant numbers 2038556 and 2337243.  Support for this research was also provided by the University of Wisconsin-Madison Office of the Vice Chancellor for Research and Graduate Education with funding from the Wisconsin Alumni Research Foundation.}\hspace{.2cm}\\
    Department of Statistics, University of Wisconsin-Madison\\
    and \\
    Jessi Cisewski-Kehe \\
    Department of Statistics, University of Wisconsin-Madison}
  \maketitle
} \fi

\if1\blind
{
  \bigskip
  \bigskip
  \bigskip
  \begin{center}
    {\LARGE\bf MaxTDA: Robust Statistical Inference for Maximal Persistence in Topological Data Analysis}
\end{center}
  \medskip
} \fi

\bigskip
\begin{abstract}
Persistent homology is an area within topological data analysis (TDA) that can uncover different dimensional holes (connected components, loops, voids, etc.) in data.  The holes are characterized, in part, by how long they persist across different scales.  Noisy data can result in many additional holes that are not true topological signal.  Various robust TDA techniques have been proposed to reduce the number of noisy holes, however, these robust methods have a tendency to also reduce the topological signal.
This work introduces Maximal TDA (MaxTDA), a statistical framework addressing a limitation in TDA wherein robust inference techniques systematically underestimate the persistence of significant homological features. MaxTDA combines kernel density estimation with level-set thresholding via rejection sampling to generate consistent estimators for the maximal persistence features that minimizes bias while maintaining robustness to noise and outliers.  We establish the consistency of the sampling procedure and the stability of the maximal persistence estimator. The framework also enables statistical inference on topological features through rejection bands, constructed from quantiles that bound the estimator's deviation probability. MaxTDA is particularly valuable in applications where precise quantification of statistically significant topological features is essential for revealing underlying structural properties in complex datasets. Numerical simulations across varied datasets, including an example from exoplanet astronomy, highlight the effectiveness of MaxTDA in recovering true topological signals.
\end{abstract}

\noindent%
\vfill

\newpage
\spacingset{1.75} 

\section{Introduction}
\label{sec:intro}

To analyze the underlying structure of complex datasets, topological data analysis (TDA) utilizes tools from algebraic topology to study the shape and connectivity of data across multiple scales. Central to TDA is persistent homology, which analyzes data through a filtration (i.e., a sequence of nested topological spaces) derived from the data, and computes homological invariants across different scales \citep{edelsbrunner2000topological,edelsbrunner2022computational} (see Section \ref{sec:prelim-topology} for more details). In this work, we refer to these invariants, which include connected components, loops, and other higher-dimensional holes, as features. By tracking when these features appear (their birth) and disappear (their death) as the filtration parameter changes, persistent homology identifies the features that persist over a range of scales. These features are represented in persistence diagrams as points with coordinates corresponding to their birth and death times where features that persists at larger scales may correspond to true signal, while the lower persistence features may be attributed to noise \citep{fasy2014confidence}. This process has broad applications. For example, in material science, it reveals significant voids that inform properties like permeability and strength \citep{robins2011theory}; in signal processing, it uncovers persistent circular features from time-delay embeddings that highlight underlying periodic dynamics \citep{perea2015sw1pers,dakurah2024subsequence}; and in astronomy, it distinguishes important cosmic structures such as clusters, filaments, and voids from noise, aiding in constraining the cosmological model \citep{pranav2017topology,xu2019finding}.

Identifying statistically significant features, particularly, the most persistent, or maximal persistent ones is challenging because persistence diagrams lack a canonical vector space structure, meaning operations like addition, averaging, and other conventional statistical techniques are not naturally defined. This difficulty is further compounded by noisy data. Methods such as kernel smoothing, developed within robust topological analysis \citep{fasy2018robust,anai2020dtm}, are employed to mitigate noise but also often reduce the lifetimes (i.e., persistences) of the maximal persistent features. The systematic underestimation of the lifetimes of these features is an artifact of the smoothing mechanisms typically employed in these robust methods. To enable statistical inference for maximal persistent features, it is helpful to address these limitations. This inference challenge arises from the need to quantify uncertainty in the presence of perturbations, such as noise, outliers, or density variation in a random sample $\mathbb{X}_n = \{\mathbf{x}_1, \cdots, \mathbf{x}_n\}$ drawn from a probability distribution $\mathbb{P}$ with compact support $\mathbb{X}$ in a space $\mathcal{X} \subset \mathbb{R}^d$. Robust topological tools aim to recover the topology of $\mathbb{X}$ by defining a smoothing function $\phi: \mathcal{X} \to \mathbb{R}$. This function, commonly a kernel density estimate (KDE), kernel distance, or distance-to-a-measure (DTM) function, is parameterized to suppress noise or reweight outliers \citep{chazal2011geometric, fasy2014confidence,fasy2018robust,anai2020dtm}. A preferred outcome would maintain high persistence for true features while reducing noise features to negligible persistence levels.

The motivation for this work is to develop an inference method that builds on these robust methods, while mitigating the reduction in the persistence of the features, in order to enhance a feature's statistical significance. The proposed framework, ``Maximal TDA'' (MaxTDA), mitigates this reduction by first estimating a KDE over the sample as an intermediate representation of the data sampling distribution. Then an upper-level set is defined for a carefully selected density threshold, and rejection sampling is used to draw samples from the thresholded KDE for subsequent statistical inference on the maximal persistent features. This process retains the robustness of the initial smoothing while producing a denser, more consistent sampling surface. Subsequent inference then involves further smoothing or directly computing a persistence diagram directly over this dense sample. This methodology is motivated by two key observations. First, the kernel smoothing enhances robustness against outliers and noise \citep{bobrowski2017topological,fasy2018robust,anai2020dtm}. Second, the thresholded KDE corrects for density variation in the sampling by providing for a denser and statistically consistent sampling surface \citep{tsybakov1997nonparametric,singh2009adaptive}, a characteristic that is crucial for maintaining the persistence of the features.
\begin{figure}[ht!]
    \includegraphics[width=\textwidth]{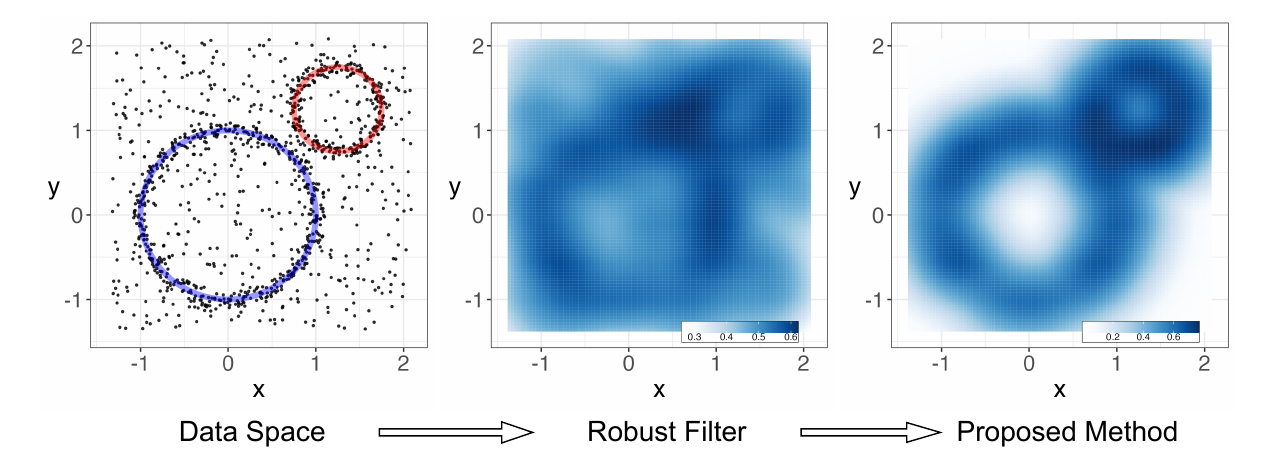}
    \caption{Illustration of the MaxTDA framework. For a data space (left), robust TDA methods applies a robust filter(e.g., KDE) to the data (middle). MaxTDA extends this by sampling from a thresholded KDE (right), enhancing robustness to noise and creating a denser sampling surface.}
    \label{fig:ga}
\end{figure}
This is illustrated in Figure \ref{fig:ga}, where the aim is to recover and maintain the persistence of key features such as the two loops (the red and blue circles) indicated by dense clusters.

The proposed MaxTDA approach presents a robust, consistent, and less biased estimator of the most persistent features in certain homology groups, which are groups that identify different dimensional holes in data; more details are presented in Section~\ref{sec:prelim-topology}. While KDEs have been used for robust persistent homology, we show that the resulting homology estimates do not preserve the strength of the true features. In Theorem~\ref{thm:consistency-subsamples}, we show that the proposed sampling technique is consistent, and in Lemma~\ref{lemma:maximal-persistence-stability}, we prove the stability of the resulting maximal persistence estimator.  From a statistical perspective, we establish that MaxTDA produces estimates with reduced bias and enhanced statistical significance. 
The remainder of this paper is structured as follows: Section~\ref{sec:prelim-topology} provides background on persistent homology, and kernel density and level-set estimation. Section~\ref{sec:methods} discusses theoretical results including consistency and bias analyses, and the statistical significance of the maximal persistence estimator. Section~\ref{sec:numerical-simulation} and ~\ref{sec:application} demonstrate the effectiveness of the MaxTDA through numerical simulations, including one motivated by a statistical challenge in exoplanet astronomy. Section~\ref{sec:conclusion} closes with implications and potential extensions of our work.

\section{Background}
\label{sec:prelim-topology}
A main constructs in TDA for studying shape features of data is persistent homology \citep{edelsbrunner2000topological,edelsbrunner2008persistent}. It provides a multi-scale characterization of the homological features of a topological space by tracking of when these features first appear (birth) and when they disappear (death) in a filtration. This section provides a review of some building blocks of persistent homology.

\paragraph{Definitions and Notations.} 
Suppose we observe a random sample $\mathbb{X}_n = \{\mathbf{x}_1, \cdots, \mathbf{x}_n \}$, drawn according to the probability distribution $\mathbb{P}$ having density function $f$ defined on the d-dimensional  compact set $\mathcal{X}$. Assume the distribution $\mathbb{P}$ is supported on the set $\text{supp}(\mathbb{P}) = \mathbb{X} \subset \mathcal{X}$. Define the Euclidean distance of any point to the set $\mathbb{A}\subset \mathcal{X}$ as:
\begin{equation}
    \text{d}_{\mathbb{A}}(\mathbf{x}) = \inf_{\mathbf{y} \in \mathbb{A}}|| \mathbf{y} - \mathbf{x} ||_2.
    \label{eqn:euclidean-distance}
\end{equation}
Let $\phi$ be any real-valued function, where $\phi: \mathcal{X} \xrightarrow{} \mathbb{R}$. We define the {\em lower-level sets} of $\phi$ as $\{\mathbf{x}: \phi(\mathbf{x}) \le \lambda\}$ and the {\em upper-level sets} of $\phi$ as $\{\mathbf{x}: \phi(\mathbf{x}) \ge \lambda\}$. In more specific settings, we let the function $\phi$ be defined on the metric space $(\mathcal{X}, \text{d}_{\mathcal{X}})$. Define the $\text{reach}(\mathbb{A})$ as the largest radius $r$, such that each point in $\cup_{\mathbf{x} \in \mathbb{A}} \text{B}(\mathbf{x}, r)$ has a unique projection unto $\mathcal{X}$, where $\text{B}(\mathbf{x}, r)$ is a ball with radius $r$ centered on $\mathbf{x}$. The reach is also referred to as the ``condition number,'' and it quantifies the smoothness of the underlying manifold \citep{federer1959curvature,niyogi2008finding}.
Denote by $\mathcal{K}(\mathcal{X}, \kappa)$ the class of all manifolds such that for $\mathbb{A} \in \mathcal{K}(\mathcal{X}, \kappa)$, $\text{reach}(\mathbb{A}) \ge \kappa$, where $\kappa$ is a fixed positive constant. Let the lower bound $\underline{b}(\mathcal{K}(\mathcal{X}, \kappa))$ and the upper bound $\overline{b}(\mathcal{K}(\mathcal{X}, \kappa))$ be positive constants depending on the geometry of the class $\mathcal{K}(\mathcal{X}, \kappa)$ but not on any specific manifold in $\mathcal{K}(\mathcal{X}, \kappa)$. 
\begin{assumption}
    The following assumptions are made for the density function $f$ and the distribution $\mathbb{P}$: (i) the support $\mathbb{X}$ of the distribution $\mathbb{P}$ is bounded, and (ii) $f$ is tame and  satisfies the following:
            $0 < \underline{b}(\mathcal{K}(\mathcal{X}, \kappa)) \le \inf_{\mathbf{x} \in \mathcal{X}}f(\mathbf{x}) \le \sup_{\mathbf{x} \in \mathcal{X}}f(\mathbf{x}) \le \overline{b}(\mathcal{K}(\mathcal{X}, \kappa)) < \infty.$
        The tameness of $f$ implies it has a finite number of critical values, ensuring the topological complexity of its level sets remains systematically bounded \citep{edelsbrunner2022computational}.
    \label{assump:density-smothness}
\end{assumption}

\subsection{Homology of simplicial complexes} \label{sec:homology_simplex}
Homology is an area of mathematics that looks for holes in a topological space, and persistent homology looks for holes in data. These holes are formalized through concepts from algebraic topology and are represented by homology groups of varying dimensions \citep{hatcher2002algebraic,edelsbrunner2022computational}. 
Specifically, the zero-dimensional homology group (${H}_0$) contains connected components (clusters), the one-dimensional homology group (${H}_1$) contains loops, the two-dimensional homology group ($H_2$) contains voids like the interior of a balloon, and more generally, the $k$-dimensional homology group (${H}_k$) represents $k$-dimensional holes.
In this work, we mainly represent topological spaces with simplicial complexes.
A $k$-simplex $\mathbf{s} = (s_0, \cdots, s_k)$ is a $k$-dimensional polytope of $k+1$ affinely independent points $s_0, \cdots, s_k$.
A simplicial complex $\mathbb{C}$ is a finite set of simplices such that for any $\mathbf{s}^{1}, \mathbf{s}^{2} \in \mathbb{C}$, $\mathbf{s}^1 \cap \mathbf{s}^2$ is a face
of both simplices, or the empty set; and a face of any simplex $\mathbf{s}\in \mathbb{C}$ is also a simplex in $\mathbb{C}$. (A face of a simplex is the convex hull of any non-empty subset of points that define the simplex.) The homology is computed from these simplicial complexes built along a sequence of filtration values.

\subsection{Persistent homology on point clouds}
The underlying topological space is often only indirectly observed through noisy point cloud data sampled from it.
A common approach to constructing simplicial complexes in TDA for point clouds is the Vietoris-Rips $(VR)$ complex \citep{vietoris1927hoheren,edelsbrunner2022computational}. A $VR$ complex is constructed over a finite set of points $S = \{s_0, s_1, \cdots, s_n\}$ using a distance parameter $\delta$. For any subset of $k$ points $\{ s_{i_1}, \cdots, s_{i_k} \}$, a $(k-1)$-dimensional simplex is formed when the pairwise Euclidean distance between all points is at most $\delta$. A collection of all such simplices forms the $VR$ complex denoted as $VR(S, \delta)$. The composition of the simplicial complex progresses hierarchically with the distance parameter $\delta$. This leads to the concept of filtration, which defines an inclusion relation between the simplicial complexes for a set of $\delta$ values. More formally, for an ordered sequence of $\delta$ values: $0 < \delta_1 < \delta_2 < \cdots < \delta_q < \infty$, the $VR$ complexes admit a nested structure as
$VR(S, 0) \subset VR(S, \delta_1) \subset \cdots \subset VR(S, \delta_q) \subset VR(S, \infty)$.
The inclusion relation between the VR complexes induces a map between the $k$-dimensional homology groups as 
${H}_k(VR(S, 0)) \xrightarrow{} {H}_k(VR(S, \delta_1)) \xrightarrow{} \cdots \xrightarrow{} {H}_k(VR(S, \infty))$.
The notion of persistent homology is developed through these homology maps by tracking the changes in the features (i.e., homology group generators) of these nested homology groups. 
The birth time and death time of features along this sequence encodes topological changes in the groups. For a homology group ${H}_k$, we denote the birth time and death time of the $j$-th feature by $b_j$ and $d_j$, respectively. The persistence of the feature is given by $d_j - b_j$, and longer persistence often is considered to be topological signal while shorter persistence often represents topological noise \citep{fasy2014confidence}. If we let $k_j$ to be the homology group dimension of the $j$-th feature, and ${J}$ the index set of the features of the homology groups, then the set 
\begin{equation}
    \text{Dgm}(S) = \{(b_j, d_j, k_j): \forall j \in J\} \cup \Delta,
\end{equation}
where $\Delta$ represent a set of points where the birth time is equal to the death time, characterizes the persistence of the features, and is used to construct a graphical summary referred to as a {\em persistence diagram}.
Figure \ref{fig:vietoris-rips-filtration} illustrates the main concepts in this section, where black points (zero-simplices) in Figure~\ref{fig:point-cloud-delta08} and~\ref{fig:point-cloud-delta15} denotes the data with cyan balls of diameter $0.8$ and $1.5$, respectively. Figure \ref{fig:persistence-diagram} shows the corresponding persistence diagram.
\begin{figure}
     \centering
      \begin{subfigure}[b]{0.24\textwidth}
         \centering
         \includegraphics[width=\textwidth]{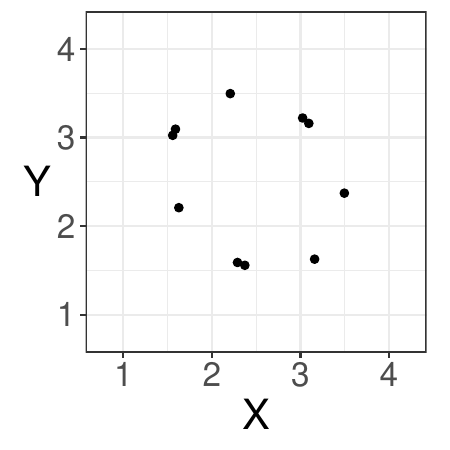}
         \caption{Point set $S$}
         \label{fig:toy-point-cloud}
     \end{subfigure}
     \hfill
     \begin{subfigure}[b]{0.24\textwidth}
         \centering
         \includegraphics[width=\textwidth]{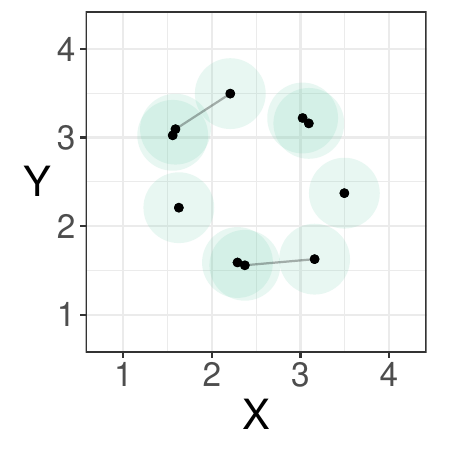}
         \caption{$VR(S, \delta=0.8$)}
         \label{fig:point-cloud-delta08}
     \end{subfigure}
     \hfill
     \begin{subfigure}[b]{0.24\textwidth}
         \centering
         \includegraphics[width=\textwidth]{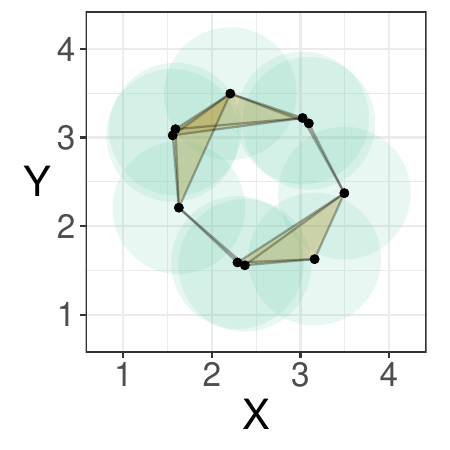}
         \caption{$VR(S, \delta=1.5$)}
         \label{fig:point-cloud-delta15}
     \end{subfigure}
     \hfill
     \begin{subfigure}[b]{0.24\textwidth}
         \centering
         \includegraphics[width=\textwidth]{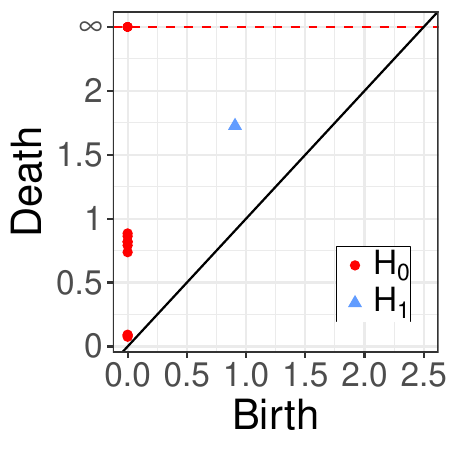}
         \caption{Persistence diagram}
         \label{fig:persistence-diagram}
     \end{subfigure}
        \caption{VR filtration and persistence diagram. The zero-simplices (black points, a-c) sampled randomly around a circle. Balls (cyan) of diameter $\delta=0.8$ and $\delta=1.5$ are drawn around the points in (b) and (c), respectively, resulting in one-simplices (black segments) and two-simplices (orange triangles). The persistence diagram (d) has $H_0$ (red points) and $H_1$ (blue triangles) features.}
        \label{fig:vietoris-rips-filtration}
\end{figure}

A set of persistence diagrams $\{\text{Dgm}(S)\}$ can be endowed with a distance measure, such as the bottleneck distance. The bottleneck distance gives the minimal $L_{\infty}$ distance between bijections of any two diagrams. Let $S_1$ and $S_2$ be two  finite compact subsets of $\mathbb{R}^{d}$, with $\text{Dgm}(S_1)$ and $\text{Dgm}(S_2)$ as their corresponding VR filtration persistence diagrams. The bottleneck distance, $\text{d}_\text{B}$, between the two persistence diagrams is defined as: 
\begin{equation}
    \text{d}_\text{B}\left(\text{Dgm}(S_1), \text{Dgm}(S_2)\right) = \inf_{\gamma} \sup_{\mu \in \text{Dgm}(S_1)} || \mu - \gamma(\mu) ||_\infty,
    \label{eqn:bottleneck-dist}
\end{equation}
where the infimum is taken over all bijections $\gamma: \text{Dgm}(S_1) \xrightarrow{} \text{Dgm}(S_2)$. Let $S_1$ and $S_2$ be endowed with the Euclidean metric (Equation~\eqref{eqn:euclidean-distance}), then their Hausdorff distance, $\text{d}_\text{H}$, is given by
\begin{equation}
    \text{d}_\text{H}(S_1, S_2) = \max \left\{ \sup_{\mathbf{s}_1 \in S_1 }\text{d}_{S_2}(\mathbf{s}_1), \sup_{\mathbf{s}_2 \in S_2 }\text{d}_{S_1}(\mathbf{s}_2)\right\}.
    \label{eqn:dist-hausdorff}
\end{equation}
A fundamental result on persistence diagrams is that they are stable summaries in many settings (i.e., a small change in a point cloud results in a small change in the corresponding persistence diagram) \citep{chazal2021introduction}. This stability relation can be stated as:
\begin{equation}
    \text{d}_\text{B}\left(\text{Dgm}(S_1), \text{Dgm}(S_2)\right) \le 2\text{d}_{\text{H}}(S_1, S_2).
    \label{eqn:stability-bn}
\end{equation}
Similar results can be obtained for functions, which is discussed in the next section.

\subsection{Persistent homology on functions}
Functions defined on the vertices of the simplicial complex $\mathbb{C}$ provides another means to characterize the topology of the underlying data generating space. Let $\phi: \mathcal{X} \xrightarrow[]{} \mathbb{R}$, and assume $\phi$ is extended to the simplices of $\mathbb{C}$ such that $\phi(\mathbf{s}) = \max_{0 \le i \le k} \phi(s_i)$ for any simplex $\mathbf{s} = (s_0, \cdots, s_k) \in \mathbb{C}$. The sequence of complexes $\mathbb{C}_\delta = \{\mathbf{s} \in \mathbb{C}: \phi(\mathbf{s}) \le \delta\}$ creates a nested structure: $\mathbb{C}_{\delta_1} \subseteq \mathbb{C}_{\delta_2}$, $\delta_1 < \delta_2$, and defines a {\em lower-level set} filtration on $\phi$. An {\em upper-level set} filtration can be defined analogously by considering the case where $\phi(\mathbf{s}) \ge \delta$. We denote the resulting persistence diagram by $\text{Dgm}(\phi)$, such that topological feature $(b, d) \in \text{Dgm}(\phi)$ persists in the space $H_k(\phi^{-1}(-\infty, \delta))$, for $b \le \delta < d$. Similar to the VR filtration, we can endow this space of persistence diagrams with the bottleneck distance as defined in Equation~\eqref{eqn:bottleneck-dist}, and these persistence diagrams can also be shown to be stable summaries (i.e., small perturbations in the function space results in small changes in the persistence diagrams) \citep{cohen2005stability,chazal2016structure}. This results in the following bound on the bottleneck distance for two functions $\phi$ and $\psi$ under the assumption of tameness (Assumption~\ref{assump:density-smothness}):
\begin{equation}
    \text{d}_{B}\left( \text{Dgm}(\phi), \text{Dgm}(\psi) \right) \le ||\phi -\psi||_\infty,
    \label{eqn:bottleneck-stability-functions}
\end{equation}
where $||\phi -\psi||_\infty = \sup_{\mathbf{x} \in \mathcal{X}}\left|\phi(\mathbf{x}) - \psi(\mathbf{x}) \right|$.

Two such functions $\phi$ that are relevant to this work are the kernel density function $f_\sigma$ and the DTM function $d_{P, m}$. The kernel density function $f_\sigma(\mathbf{x})$ with bandwidth $\sigma$, and its empirical estimate $\widehat{f}_\sigma(\mathbf{x})$ are defined as:
\begin{equation}
    f_\sigma(\mathbf{x}) = \int_{\mathcal{X}}  K_\sigma(||\mathbf{x} - X||_2)d\text{P}(X), \quad \widehat{f}_\sigma(\mathbf{x}) = \frac{1}{n} \sum_{i = 1}^n K_\sigma(||\mathbf{x} - \mathbf{x}_i||_2),
    \label{eqn:kernel-density}
\end{equation}
where $K_\sigma = {\sigma^{-d}}K\left(||\mathbf{x} - X||_2/\sigma\right)$, and $K$ is a d-dimensional kernel that is non-negative and integrates to one.
While this kernel density function captures the shape and distribution of mass in the space $\mathcal{X}$, the DTM function provides a robust means to characterize this shape by approximating its distance function. The DTM function $d_{\text{P}, m}(\mathbf{x})$ and its empirical estimate $\widehat{\text{d}}^2_{\text{P}, m}(\mathbf{x})$ are defined as \citep{chazal2011geometric}:
\begin{equation}
    d_{\text{P}, m}(\mathbf{x}) = \sqrt{\frac{1}{m}\int_0^m F_{\mathbf{x}}^{-1}(\mu)d\mu}, \quad \quad \widehat{\text{d}}^2_{\text{P}, m}(\mathbf{x}) = \frac{1}{k}\sum_{X \in N_K(\mathbf{x})} || X - \mathbf{x} ||_2,
    \label{eqn:dtm-function}
\end{equation}
where $F_{\mathbf{x}}(t) = \text{Pr}\left( ||X - \mathbf{x}||_2 \le t  \right)$, $0 < m < 1$ is the resolution, and $N_K(\mathbf{x})$ is the set of $k$-nearest neighbors ($k$-NNs) to $\mathbf{x}$. The DTM filtration is a robust approximation of the VR filtration. 

\subsubsection{Sensitivity and artifacts in filtration functions}

In practice, the estimation of filtration functions relies on the empirical probability measure $\mathbb{P}_n$, which assigns a probability mass of $1/n$ to each data point $\mathbf{x}$. As a result, the empirical function $\phi_n$, whether it represents the empirical KDE or the empirical DTM function, exhibits sensitivity to noise and sample density variations. This sensitivity directly affects the persistence of the resulting features.
Recent efforts have generalized these filtration methods by defining more flexible kernel distances, enabling the use of more robust kernel distances \citep{fasy2018robust,anai2020dtm}. However, this robustness is largely achieved by smoothing out the low persistence features, which tends to reduce the persistence of all features, especially the most persistent features. Figure~\ref{fig:overlap-annulus}  and \ref{fig:overlap-annulus-smoothed} illustrates this phenomenon. This work seeks to build on these robust filtration schemes while minimizing the reduction in the life span of the most persistent feature. We propose a rejection sampling technique on a smooth surface, and demonstrate that the sampled points retain the same features as  $\mathbb{X}$ and that this approach minimally affects the persistence of the most prominent features. This results in a method with reduced bias in estimating maximal persistence. Additionally, we derive lower bounds for maximal persistence, providing a principled approach to assess the statistical significance of the most persistent features.

 
\section{Maximal TDA method} \label{sec:methods}
This section present the MaxTDA method and corresponding theoretical results. In particular, we show the construction of the smooth subsamples for inference on the maximal persistence features, prove its homology preserving properties as well as some bias reduction and consistency results.

\subsection{Overview of methodology}
Before presenting the technical details of our method, we first provide an overview of our approach to estimating and performing statistical inference on the maximal persistent features. The key challenge we address is how to reliably estimate the most persistent features from noisy point cloud data while minimizing the persistence reduction (see Section~\ref{subsec:statistical-motivation}). Traditional approaches often face a trade-off between noise reduction and feature preservation. Our method seeks to resolve this through the construction of ``smooth subsamples'' using a combination of kernel density and level-set estimation via rejection sampling, where we only accept proposed points where the estimated density exceeds a threshold $\lambda$. The details of this construction are discussed in Section~\ref{subsec:subsamples}, and the validity of such a construction is established in Theorem~\ref{thm:consistency-subsamples}.
This thresholding naturally filters out likely noise points while preserving the strength of genuine features, as true features tend to manifest in regions of high density.

The remainder of this section develops statistical inference methods for working with features constructed from the smooth subsamples of the thresholded KDE. We analyze the bias reduction properties of the maximal persistence estimator and develop methods for assessing its statistical significance, providing both theoretical guarantees and practical tools for identifying statistically significant features in noisy data with varying density distributions. Extensions of these inference methods to functions of the maximal persistence are also discussed.

\subsection{Stability of Maximal Persistence} \label{subsec:statistical-motivation}
The support $\mathbb{X}$ is not directly observed but is studied through the point cloud $\mathbb{X}_n$. Let $\text{Dgm}(\mathbb{X}_n)$ be the VR-based filtration persistence diagram on $\mathbb{X}_n$, and $\text{Dgm}(\mathbb{X})$ the true underlying persistence diagram on the support $\mathbb{X}$. For $\phi$ defined on $\mathcal{X}$, and its empirical estimate $\phi_n$, let $\text{Dgm}(\phi(\mathbb{X}))$, and $\text{Dgm}(\phi_n(\mathbb{X}))$ denote their persistence diagrams from upper-level set filtrations of $\phi$ and $\phi_n$, respectively. When an exposition applies to either a VR filtration or the upper-level set filtration of the KDE or DTM, the persistence diagram is generically denoted as $\text{Dgm}(\cdot)$ or simply $\text{Dgm}$. Define the maximum persistence of the features on the persistence as $\text{mp}\left[ \text{Dgm}(\cdot) \right]$. The following lemma presents the stability result for the maximal persistence estimator.
 \begin{lemma}[Maximal Persistence Stability]
     Let $\Delta$ be the persistence diagram with points only along the diagonal. Let $\phi_n$ be an empirical KDE or DTM function defined on the sample $\mathbb{X}_n$. Then the following results hold: \\
         (i) The maximum persistence can be expressed in terms of the bottleneck distance:
         \begin{equation}
             \text{mp}\left[\text{Dgm}(\cdot)\right] = 2\text{d}_{\text{B}}\left(\text{Dgm}(\cdot), \Delta  \right).
         \end{equation}
         (ii) Let $\widehat{\nabla}$ be defined as $\widehat{\nabla} = \left| \text{mp}[\widehat{\text{Dgm}}] - \text{mp}[\text{Dgm}] \right|$, then it holds that:
         \begin{equation}
             \widehat{\nabla} \le 2\text{d}_{\text{B}}\left( \widehat{\text{Dgm}}, \text{Dgm} \right),
         \end{equation}
         where $\widehat{\text{Dgm}}$ denotes the empirical persistence diagram estimate of $\text{Dgm}$.
     \label{lemma:maximal-persistence-stability}
 \end{lemma}
\begin{proof}
    The maximum persistence $\text{mp}[\text{Dgm}(\cdot)]$ is defined as: $\text{mp}[\text{Dgm}(\cdot)] = \max_{(b, d)\in \text{Dgm}}|d - b|$.
    Similarly, the bottleneck distance $\text{d}_{\text{B}}(\text{mp}[\text{Dgm}(\cdot)], \Delta)$ is defined as:
    \begin{equation}
        \text{d}_{\text{B}}(\text{mp}[\text{Dgm}(\cdot)], \Delta) = \inf_{\gamma} \sup_{(b, d)\in \text{Dgm}}|| (b, d) - \gamma((b, d)) ||_\infty,
    \end{equation}
    where $\gamma: \text{Dgm} \xrightarrow{} \Delta$ defines a bijection between $\text{Dgm}$ and $\Delta$. Note that since $\Delta$ is the diagonal, the optimal bijection $\gamma$ is the orthogonal projection of points in $\text{Dgm}(\cdot)$ to $\Delta$, hence $\gamma((b, d)) = \left( \frac{b+d}{2}, \frac{b+2}{2} \right)$. It then follows that:
    \begin{equation}
        \text{d}_{\text{B}}(\text{mp}[\text{Dgm}(\cdot)], \Delta) = \sup_{(b, d) \in \text{Dgm}(\cdot)} \left\lVert (b, d) -  \left( \frac{b+d}{2}, \frac{b+d}{2} \right)\right\rVert_\infty = \frac{|d^\prime-b^\prime|}{2},
    \end{equation}
    where $(b^\prime, d^\prime)$ are the birth-death pair with the maximal persistence.
    The bound for $\widehat{\nabla}$ can be derived based on the expression:
    \begin{equation}
        \widehat{\nabla} = \left| \text{mp}[\widehat{\text{Dgm}}] - \text{mp}[\text{Dgm}] \right| = 2\left| \text{d}_{\text{B}}(\text{mp}[\widehat{\text{Dgm}}], \Delta) - \text{d}_{\text{B}}(\text{mp}[{\text{Dgm}}], \Delta)  \right| \le 2\text{d}_{\text{B}}(\text{mp}[\widehat{\text{Dgm}}], \text{mp}[{\text{Dgm}}]),
    \end{equation}
    where the last inequality follows from the reverse triangle inequality for metrics. 
\end{proof}
The main object of interest in this work is the maximal persistence $\text{mp}\left[\text{Dgm}(\cdot)\right]$. We now describe the framework to consistently estimate it while reducing the associated bias inherent in estimating these maximal values.

\subsection{Smooth sampling surface} \label{subsec:subsamples}
The methodology for constructing the smooth sampling surface that maximizes the persistence of features is described next. This approach can be used to either maximize the persistence of a single feature or multiple features, depending on the application. The goal here is to obtain samples that better approximate the true underlying topology, and these samples are subsequently used for inference on the maximal persistence features. Our approach uses kernel density estimation to create a smooth sampling surface, enabling the generation of samples that preserve the underlying topological structure. Specifically, given the observed data $\mathbb{X}_n$ drawn according to the distribution $\mathbb{P}$ with density $f$, we approximate this density with the KDE estimate $\widehat{f}_\sigma$. This provides a smooth surface that captures the structure of the manifold while reducing the noise. A dense sample $\mathbb{X}_n^*$ is drawn from the smooth surface using rejection sampling \citep{devroye1986non}. The $\mathbb{X}_n^*$ serves to preserve the persistence of the homology features relative to $\mathbb{X}_n$. Rejection sampling requires a target distribution and a proposal distribution, where samples are drawn from a proposal distribution because of difficulties sampling from the target distribution. For the purpose of this work, the proposal distribution $\mathbb{Q}$ is a function of the volume enclosing the topological space $\mathcal{X}$. The target distribution is the KDE $\widehat{f}_\sigma$. The objective then is to draw samples $\mathbf{x}^*$ according to $\mathbb{Q}$ and accept them based on the target density $\widehat{f}_\sigma$. In particular, for some $\Gamma \ge \sup_{\mathbf{x^*} \in \mathcal{X}} \widehat{f}_\sigma(\mathbf{x}^*)$, the sample $\mathbf{x}^*$ is accepted with probability $\widehat{f}_\sigma(\mathbf{x}^*)/\Gamma$. Algorithm \ref{alg:subsampling} outlines the sampling scheme described here.
\begin{algorithm}
\caption{Smooth Subsampling}
\label{alg:subsampling}
\begin{algorithmic}[1]
\Require Observed data $\{ \mathbf{x}_1, \cdots, \mathbf{x}_n \}$, density threshold $\lambda$, number of generated points $B$.
\Stepa Fit the KDE $\widehat{f}_\sigma$ to the data sample $\{ \mathbf{x}_1, \cdots, \mathbf{x}_n \}$.
\Stepb For $k$ in the range $1, \cdots, B$, \textbf{do} \textbf{Step 3} to \textbf{Step 4}. 
\Stepc Compute $\mathbf{x}^*$ as follows:
\Statex \quad \textbf{Repeat:}
\Statex  \quad \quad (i) Draw a sample ${\mathbf{x}^*}$ from the proposal distribution $\mathbb{Q}$.
\Statex \quad \quad (ii) Compute the density associated with the sample $\mathbf{x}^*$, i.e., evaluate $\widehat{f}_\sigma(\mathbf{x}^*)$.
\Statex \quad \quad (iii) Sample a point $u \sim U(0, \Gamma)$.
\Statex \quad \textbf{Until:} $ u \le \widehat{f}_\sigma(\mathbf{x}^*)$, and $\widehat{f}_\sigma(\mathbf{x}^*) \ge \lambda$.
\Stepd Set $\mathbf{x}^*_k = \mathbf{x}^*$.
\Output Return the samples $\mathbb{X}_B^* = \{ \mathbf{x}^*_1, \cdots, \mathbf{x}^*_B\}$.
\end{algorithmic}
\end{algorithm}
The resulting sample is used to construct a distribution of maximal persistence values by generating random persistence diagrams from the transformed data space. A threshold $\lambda$ is then selected to maximize the persistence of the targeted prominent features.

We now show that the samples $\mathbb{X}_n^*$ obtained via Algorithm \ref{alg:subsampling} preserves the homology of $\mathbb{X}$. This largely follows from the theory of level-set estimation, especially the work of \cite{cuevas1997plug}. An interesting observation made in \cite{bobrowski2017topological} is that recovering the homology of $\mathbb{X}$ does not rely on the consistency of the KDE $\widehat{f}_\sigma$. Also, to avoid making assumptions on the shape of the space $\mathbb{X}$, the Hausdorff metric is used to measure the closeness of the approximation.
\begin{theorem}[Convergence of Smooth Subsamples]
    Let $\mathbb{P}$ be compactly supported on the set $\mathbb{X}$, having bounded density $f$ and $f>\lambda$ for some positive constant $\lambda$. Assume the kernel function $K_\sigma$ is a decreasing function of $\mathbf{x}$ such that as $||\mathbf{x}||\xrightarrow{} \infty$, we have $||\mathbf{x}||^{d+1}K_\sigma(||\mathbf{x}||) \xrightarrow{} 0$. Further assume that $K_\sigma$ is a bounded density such that for some $r_1, r_2$, 
    \begin{equation}
        K_\sigma(||\mathbf{x}||) \ge r_1 \mathbbm{1}(\mathbf{x} \in B(0, r_2)).
    \end{equation}
    Let $\beta_n$ be of order $o(n/\log n)^{1/d}$, then
        $\beta_n \text{d}_\text{H}(\mathbb{X}_n^*, \mathbb{X}) \xrightarrow{} 0$ a.s., where $\beta_n \sigma$ goes to zero with $n$ large.
   \label{thm:consistency-subsamples}
\end{theorem}
In our analysis, we used a Gaussian kernel, which is not compactly supported. Hence we make the additional assumption that the bandwidth $\sigma \xrightarrow{} 0$ as $n\xrightarrow{} \infty$ in such a way that  $\beta_n^{d+1}\sigma \xrightarrow{} 0 $, then the proof follows directly from Theorem $3$ in \cite{cuevas1997plug}. 
By constructing a smooth subsample, we intrinsically reduce the magnitude of any topological error in subsequent persistent homology computations on these smooth subsamples. For example, in the initial sample $\mathbb{X}_n$, the randomness in the sample could introduce points that results in additional features. The kernel smoothing initially reduces the presence of such outlying points, and an appropriate choice of $\lambda$ (which depends on the specific application) enhances the originally significant homological features. {\em In summary, unless the randomness introduces features that dominate the most persistent real feature, MaxTDA guarantees a smooth recovery of this original dominant feature}. If randomness in the sample introduces more persistent features than the most persistent real feature, it is not generally feasible to recover the real feature \citep{fasy2014confidence,bobrowski2017topological}. We demonstrate this concept in our numerical studies in Section~\ref{subsec:topological-recovery} and~\ref{subsec:sampling-variability}. The next section discusses how to select the optimal smoothing bandwidth and the level-set threshold parameters.

\subsubsection{Parameter selection}\label{sec:parameter-selection}
The choice of KDE bandwidth $\sigma$ and the level-set threshold $\lambda$ are essential to $\mathbb{X}_n^*$ recovering the topology of $\mathbb{X}$. These values are not known in practice, hence we provide a data-dependent estimation process for selecting these parameters. For a given homology dimension, let $\ell_i(\lambda, \sigma)$ be the lifetime (i.e., persistence) of the $i$-th feature on the persistence diagram $\widehat{\text{Dgm}}$ associated with $\mathbf{X}_n^*$. Consider the  ordered lifetimes
%
    $\ell_1(\lambda, \sigma) \ge \ell_2(\lambda, \sigma) \ge \cdots \ge \ell_T(\lambda, \sigma)$,
%
where $T$ is the number of features of interest. The cumulative persistence of the top $T$ features is given by:
$C_T(\lambda, \sigma) = \sum_{i = 1}^T \ell_i(\lambda, \sigma)$. The goal is to choose the parameter $\lambda$ and $\sigma$ that maximizes $C_T(\lambda, \sigma)$. The parameters $(\lambda^*, \sigma^*)$ are determined by solving the optimization problem: $(\lambda^*, \sigma^*) = \argmax_{(\lambda, \sigma) \in \Theta} C_T(\lambda, \sigma)$, where $\Theta$ denotes the feasible parameter space for $\lambda$ and $\sigma$. Note that this process can be augmented to emphasize certain features by assigning weights $\{\omega_1, \cdots, \omega_T\}$ to the lifetimes. The number of features $T$ can be chosen based on the expected topology of $\mathbb{X}$, or by adaptively by analyzing the decay of the ordered lifetimes $\ell_i(\lambda, \sigma)$. A sharp drop in $\ell_i(\lambda, \sigma)$ beyond a certain index indicates a natural cutoff for $T$. For the bandwidth $\sigma$, we found that the average $k$-NN distance (for $k$ between $1$ and $5$) between points in $\mathbb{X}_n$ provides a good parameter search space.

\subsection{Bias reduction}
Existing methods for estimating a persistence diagram in the presence of noise or sampling variability can identify the maximal persistent features. This often involves smoothing out low persistence features, which consequently reduces the lifetime of the most persistent $H_1$ features. This results in a bias in estimating the lifetime of the maximal persistent features. In this section, we discuss this phenomenon and provide results that shows the proposed MaxTDA method helps reduce this bias for an appropriate choice of thresholding parameter $\lambda$ and for a range of bandwidths $\sigma$.

\subsubsection{Source of bias in maximal persistence}
Robust persistent homology methods, such as smoothing, subsampling, filtering, or thresholding, implicitly bias the persistence estimates by reducing the lifetimes of the features. The following example illustrates this bias. Let $\mathcal{P}$ be a class of probability distributions satisfying Assumption~\ref{assump:density-smothness}. Further assume that there exists positive constants $c$ and $c^\prime$ such that for data $\mathbf{x}\in \mathbb{X}$ and $d^\prime < d$:
\begin{equation}
    \text{vol}_d\left(B(\mathbf{x}, r) \cap \mathbb{X}\right) \ge c\left(1- \frac{r^2}{4\kappa^2}\right)^{d^\prime/2}r^{d^\prime} \ge c^\prime r^{d^\prime},
\end{equation}
where $\text{vol}_{d}(\cdot)$ denotes the volume of a $d$-dimensional ball, and $\kappa$ is as defined in Assumption~\ref{assump:density-smothness}.
This is the usual regularity assumption that removes certain pathological manifolds, such as those with sharp peaks or cusps. In practical terms, for every point $\mathbf{x} \in \mathbb{X}$, if you take a ball of radius $r$ around $\mathbf{x}$, the portion of the ball lying in $\mathbb{X}$ has a volume that scales with $r^{d^\prime}$, that is, $\mathbb{X}$ is ``thick enough'' in every small neighborhood such that there are no parts that are infinitesimally thin or sharply peaked. Let $\mathbb{X}_n$ be drawn according to a distribution $\mathbb{P} \in \mathcal{P}$ which is supported on $\mathbb{X}$. Let $\widehat{\text{Dgm}}$ and Dgm be the persistence diagrams associated with $\mathbb{X}_n$ and $\mathbb{X}$, respectively. Then the following inequality holds:
\begin{equation}
    \text{mp}[\widehat{\text{Dgm}}] \le \text{d}_\text{B}(\widehat{\text{Dgm}}, \text{Dgm}) + \text{mp}[\text{Dgm}],
    \label{eqn:mp-decomp}
\end{equation}
which follows from applying the triangle inequality to $\text{mp}[\widehat{\text{Dgm}}] =  \text{d}_\text{B}(\widehat{\text{Dgm}}, \nabla)$.
This implies the bias: $E( \text{mp}[\widehat{\text{Dgm}}] ) - \text{mp}[\text{Dgm}]$ is directly upper bounded  by the expected bottleneck distance between $\widehat{\text{Dgm}}$ based on $\mathbb{X}_n$ and $\text{Dgm}$ based on $\mathbb{X}$. Therefore, a ``good'' representation of $\mathbb{X}$ can lead to a lower bias in estimating $\text{mp}[\widehat{\text{Dgm}}]$. Next, we discuss how the proposed framework
provides a good representation of $\mathbb{X}$ with a thresholded KDE.

\subsubsection{Role of sampling and thresholding}
Consider the setup where two samples, $\mathbb{X}_{n, 0}^*$ and $\mathbb{X}_{n, \lambda}^*$, are drawn using Algorithm~\ref{alg:subsampling} with threshold values of $0$ and  $\lambda$, respectively.  The choice of $\lambda$ and $n$ influence the bias in estimating the maximal persistence. We consider the case of the VR filtration, but the analysis also applies to filtrations of $\phi(\cdot)$. From Equation~\eqref{eqn:mp-decomp}, the maximal persistence associated with $\mathbb{X}_{n, \lambda}^*$ is given as follows:
\begin{equation}
    \zeta(n, \lambda)\text{mp}[\widehat{\text{Dgm}}(\mathbb{X}_{n, \lambda}^*)] = \text{d}_\text{B}(\widehat{\text{Dgm}}(\mathbb{X}_{n, \lambda}^*), \text{Dgm}(\mathbb{X})) + \text{mp}[\text{Dgm}(\mathbb{X})],
    \label{eqn:mp-dcomp-lambda}
\end{equation}
where $\zeta(n, \lambda)$ is an unspecified sequence depending on $n$ and $\lambda$, and goes to $1$ for $n$ large. In the limit as $\lambda\xrightarrow{}0$, we have by construction that $\zeta(n, \lambda) = \zeta(n, 0)$. Hence the bias of the smoothed and unsmoothed estimators are the same as $n \xrightarrow{}\infty$, and $\lambda \xrightarrow{} 0$:
\begin{equation}
    \lim_{n \xrightarrow{} \infty, \lambda \xrightarrow{} 0}E\left( \text{d}_\text{B}(\widehat{\text{Dgm}}(\mathbb{X}_{n, \lambda}^*), \text{Dgm}(\mathbb{X}))\right) = \lim_{n \xrightarrow{} \infty} E\left( \text{d}_\text{B}(\widehat{\text{Dgm}}(\mathbb{X}_{n, 0}^*), \text{Dgm}(\mathbb{X})\right).
\end{equation}
Under finite sampling, the benefits of the thresholding lie in the difference in the rates of convergence of both $\mathbb{X}_{n, \lambda}^*$ and $\mathbb{X}_{n, 0}^*$ to $\mathbb{X}$. For example, consider $\beta_n \xrightarrow{} \infty$ from Theorem~\ref{thm:consistency-subsamples}, \cite{cuevas1997plug} show any rate of order $(n/\log n)^{1/d} = O(\beta_n)$ cannot be achieved by $\mathbb{X}_{n, 0}^*$
That is, consider a convergent rate that is faster than $\beta_n$ for $\mathbb{X}_{n, \lambda}^*$ to $\mathbb{X}$, say $\beta_n^*\ge (n/\log n)^{1/d}$, then $\beta_n^*$ cannot be achieved when estimating $\mathbb{X}$ with $\mathbb{X}_{n, 0}^*$ \citep{cuevas1997plug}.  
Hence, for an appropriate choice of $\lambda$, we conjecture that:
\begin{equation}
    E\left( \text{d}_\text{B}(\widehat{\text{Dgm}}(\mathbb{X}_{n, \lambda}^*), \text{Dgm}(\mathbb{X})) \right) \le E\left(  \text{d}_\text{B}(\widehat{\text{Dgm}}(\mathbb{X}_{n, 0}^*), \text{Dgm}(\mathbb{X})) \right).
\end{equation}
While this inequality has been observed empirically (see Figure~\ref{fig:overlap_annulus_max_pers}), a formal theoretical proof remains an open challenge. The primary difficulty in establishing such a result lies in deriving an explicit form for $\zeta(n, \lambda)$, which would require strong assumptions on the geometric properties of the support $\mathbb{X}$ to obtain a closed-form expression.

\subsection{Statistical significance of the maximal persistence}
\label{sec:signifiance-max-persistence}
The statistical significance of the maximal persistence estimator $\text{mp}[\widehat{\text{Dgm}}]$ is determined through a lower bound for $\text{mp}[\widehat{\text{Dgm}}]$. This is equivalent to bounding the difference $\widehat{\nabla} = \left| \text{mp}[\widehat{\text{Dgm}}] - \text{mp}[\text{Dgm}] \right|$. A method for constructing confidence sets for persistence diagrams by bootstrapping the bottleneck distance was proposed in \cite{fasy2018robust}. The construction of the lower bound for $\widehat{\nabla}$ follows the same framework. We first state the following consistency result for $\widehat{\nabla}$ based on the upper-level set filtration of the density function, and similar consistency results holds for other functions such as the DTM and other kernel distances.
\begin{theorem}[Consistency]
    Let $\phi$ be a density function defined on $\mathcal{X}$, and let $\phi_n$ be its empirical estimate according to Equation~\eqref{eqn:kernel-density} based on the sample $\mathbb{X}_n^*$ from Algorithm~\ref{alg:subsampling}. Let $\{c_1, \cdots, c_k\}$ and $\{c^n_1, \cdots, c^n_k\}$ be the critical points of $\phi$ and $\phi_n$, respectively. Assume that the critical points of $\phi$ and $\phi_n$ are close enough such that the maximal difference at these critical points is bounded as: $\max_{i}| \phi_n(c_i^n) - \phi(c_i)| \le \frac{1}{2} \min_{i\ne j} | \phi(c_i) - \phi(c_j) | - || \phi_n - \phi ||_\infty$, and  $2|| \phi_n - \phi ||_\infty \leq \frac{1}{2} \min_{i\ne j} | \phi(c_i) - \phi(c_j) |$.
    Then $\text{mp}[\widehat{\text{Dgm}}(\phi_n)]$ is a consistent estimator of $\text{mp}\left[\text{Dgm}(\phi)\right]$:
    \begin{equation}
        \widehat{\nabla} = \left| \text{mp}[\widehat{\text{Dgm}}(\phi_n)] - \text{mp}[\text{Dgm}(\phi)]\right| \xrightarrow[]{P} 0, \quad \text{as } n \xrightarrow{} \infty.
    \end{equation}
\end{theorem}
\begin{proof}
    The proof follows from the regular consistency results on kernel density estimation and the critical distances lemma in \cite{devroye2001combinatorial} and \cite{fasy2018robust}.
    By the bottleneck stability theory, we have that $\text{d}_{\text{B}}(\text{Dgm}(\phi_n), \text{Dgm}(\phi)) \le || \phi_n - \phi ||_\infty$. Note that for the upper-level sets filtration of these functions, the homology only changes at the critical points. Assume that $(\phi(c_i), \phi(c_j)) \in \text{Dgm}(\phi)$ and $(\phi(c^n_i), \phi(c^n_j)) \in \text{Dgm}(\phi_n)$. Let $\gamma: \text{Dgm}(\phi_n) \xrightarrow[]{} \text{Dgm}(\phi)$ be the optimal bottleneck matching between the two diagrams. Under the assumption that 
    $\max_{i}| \phi_n(c_i^n) - \phi(c_i)| \le \frac{1}{2} \min_{i\ne j} | \phi(c_i) - \phi(c_j) | - || \phi_n - \phi ||_\infty$, which implies 
    $\min_{i\ne j} | \phi(c_i) - \phi(c_j) | - \max_{i}| \phi_n(c_i^n) - \phi(c_i)| \ge \max_{i}| \phi_n(c_i^n) - \phi(c_i)| + 2|| \phi_n - \phi ||_\infty$, it follows that $\gamma(\phi(c^n_i), \phi(c^n_j)) = (\phi(c_i), \phi(c_j))$. By Lemma~\ref{lemma:maximal-persistence-stability}, we have that $\widehat{\nabla} \le \max_{i}| \phi_n(c_i^n) - \phi(c_i)|$. 
    Define $\phi_n = \widehat{f}_{\sigma_2}$, and let $\Tilde{f}_{\sigma_1}$ be the KDE on $\mathbb{X}_n$ and $L_\lambda = \{\mathbf{x}: \Tilde{f}_{\sigma_1}(\mathbf{x}) > \lambda\}$. Observe that $\widehat{f}_{\sigma_2} = \left( f_{\mathbb{X}_n^*} \ast K_{\sigma_2}\right)$ where $(\cdot\ast\cdot)$ denotes the convolution operation, and $f_{\mathbb{X}_n^*}(\mathbf{x}) \propto \frac{ \Tilde{f}_{\sigma_1}(\mathbf{x})\mathbbm{1}(\mathbf{x} \in L_\lambda)} {\int_{L_\lambda} \Tilde{f}_{\sigma_1}(\mathbf{y})dy}$. The conclusion follows from the regular consistency assumption on the bandwidths $\sigma_1, \sigma_2\xrightarrow{} 0$ and sample size $n \xrightarrow{} \infty$ of the KDE \citep{devroye2001combinatorial}.
\end{proof}
\noindent
Next, we describe the framework for assessing the statistical significance of the maximal persistence features via a Monte-Carlo procedure.

\subsubsection{Construction of confidence sets}
It is common to consider homology features with longer persistence as topological signal \citep{fasy2014confidence}. Thus ${H}_{k > 0}$ features with longer life spans can be interpreted as being more statistically significant than those with shorter life spans. For example, in time series analysis, one method for determining periodicity examines the persistence of loops in a TDE space. A perfectly circular loop suggests an underlying periodic signal, and \cite{perea2015sw1pers} proposes estimating its period using a function of $\text{mp}[\text{Dgm}]$, though a method for quantifying the statistical significance of this estimate was not established. The proposed MaxTDA framework addresses this gap by providing tools to determine the statistical significance of such a periodicity measure.

Methods for estimating the statistical significance of the features through confidence sets were discussed in \cite{fasy2014confidence}. However, these methods bound the bottleneck distance with the Hausdorff distance or distances of functions defined on the data space, which shifts the randomness in the construction to the original data space. These bounds are not tight in many cases (e.g., \citealt{fasy2018robust,glenn2024confidence}). A method that restricts the randomness to the persistence diagrams, by directly bootstrapping the bottleneck distance was introduced in \cite{fasy2018robust}. We first describe the process for constructing confidence sets for the features with maximal persistence on the persistence diagram, which in our case amount to lower bounds for the maximal persistence.

Given significance level $\alpha \in (0, 1)$, the goal is to find $t_{\alpha}$ such that: $\text{Pr}(\text{d}_{\text{B}}(\widehat{\text{Dgm}}, \text{Dgm}) > t_\alpha) \le \alpha$ as $n \xrightarrow[]{} \infty$.
The confidence set on a persistence diagram can be constructed by considering points in $\widehat{\text{Dgm}}$ whose distance to the diagonal exceeds $t_\alpha$, $\left\{ (b, d) \in\widehat{\text{Dgm}}: |d-b| > 2t_\alpha  \right\}$.
%
%
This construction extends to the maximal persistence estimator through the relation:
\begin{equation}
    \text{Pr}(\widehat{\nabla} > 2t_\alpha) \le \text{Pr}(\text{d}_{\text{B}}(\widehat{\text{Dgm}}, \text{Dgm}) > t_\alpha) \le \alpha.
\end{equation}
There are two ways to visualize this confidence set on $\widehat{\text{Dgm} }\subset \mathbb{R}^2$. The first is to draw $\text{d}_{\text{B}}$-balls with  side length of $2t_\alpha$ centered on each point in $\widehat{\text{Dgm}}$. Then using the closeness to the diagonal, a point is considered to be be a topological noise if its $\text{d}_{\text{B}}$-ball intersects with the diagonal line. The second and equivalent option is to add a diagonal band (rejection band) of width $\sqrt{2}t_\alpha$ to $\widehat{\text{Dgm}}$, and points in $\widehat{\text{Dgm}}$ that falls within this band are elements of $ \left\{ (b, d) \in\widehat{\text{Dgm}}: |d-b| \le 2t_\alpha  \right\}$, a rejection set, and are deemed to not be statistically significant at the $\alpha$ significance level. Note that the rejection band is constructed individually for each homology dimension. The $t_\alpha$ can be estimated via a Monte-Carlo procedure described in the next section.

\subsubsection{Monte-Carlo estimation procedure}
In this section, a Monte-Carlo procedure is proposed to estimate the $t_\alpha$. Draw the sample $\mathbb{X}_n^{*(b)}$ using  Algorithm~\ref{alg:subsampling} as follows: first, take a bootstrap sample $\mathbb{X}_n^{(b)}$ from the original sample $\mathbb{X}_n$. Using Algorithm~\ref{alg:subsampling}, generate $\mathbb{X}_n^{*(b)}$ with $\mathbb{X}_n^{(b)}$ as the underlying observed sample. 
Let $\phi_n^{(b)}$ be the function associated with the sample $\mathbb{X}_n^{*(b)}$, and $\phi_n$ to be the function associated with the quantity $\mathbb{X}_n^{*}$, obtained by applying Algorithm~\ref{alg:subsampling} to $\mathbb{X}_n$. Compute the empirical quantity $\hat{t}^{(b)} = \text{d}_{\text{B}}(\widehat{\text{Dgm}}(\phi_n^{(b)}), \widehat{\text{Dgm}}(\phi_n) )$.
This process is repeated $b=1, \cdots, N$ times. Let $\widehat{\text{F}}_n$ be the empirical distribution function of this set of observations: $\{ \hat{t}^{(b)}: b=1, ..., N\}$. Let $\hat{t}_\alpha$ be the $1-\alpha$ quantile of $\widehat{\text{F}}_n$. Also let ${\text{F}}_n$ be the distribution function of the quantity $\text{d}_{\text{B}}(\widehat{\text{Dgm}}(\phi_n), {\text{Dgm}}(\phi) )$. Then the following result holds.
\begin{lemma}
The $1-\alpha$ quantile $\hat{t}_\alpha$ is a consistent estimator of $t_\alpha$, that is, 
    $\sup_t |\widehat{\text{F}}_n(t) - {\text{F}}_n(t)| \xrightarrow[]{p} 0.$
    \label{eqn:consistency-t-alpha}
\end{lemma}
\noindent
This result follows directly from Theorem 19 and Corollary 20 in \cite{fasy2018robust}. This process is used to determine the statistical significance of the maximally persistent $H_{k>0}$ features.

\begin{remark}
    The inference procedure developed in this work can be extended to additive or multiplicative transformations of $\widehat{\nabla}$. For example, the statistical significance of the normalized periodicity score $\text{mp}[\text{Dgm}]/\sqrt{3}$ can be derived through the distribution of $\widehat{\nabla}/\sqrt{3}$, which amounts to estimating the empirical quantile function $\sqrt{3}\hat{t}_\alpha$. These results can also be adapted for minimal persistence. 
\end{remark}


\section{Numerical Validations}
\label{sec:numerical-simulation}
This section presents numerical studies that demonstrate the performance of MaxTDA. First we show the quality of the topological recovery achieved by the proposed method in terms of the number of features recovered and the persistence of these features.  

\subsection{Quality of topological recovery} \label{subsec:topological-recovery}
The first numerical experiment aims to recover a densely sampled circle while treating a sparse circle as noise. The data consist of $50$ samples around a unit circle, $500$ samples around a radius-$0.5$ circle (both perturbed by N($0$, $\sqrt{0.05}$)), and $450$ uniform samples in $[-1, 1]^2$. These three samples give $\mathbb{X}_n$ with $n=1000$. The VR, DTM, and KDE persistence diagrams were computed, with DTM parameter $m=0.9$ chosen over a grid of points in the interval $(0, 1)$, and the KDE bandwidth set at $0.1$.  A complete comparison across various bandwidths is given in Section~\ref{subsec:pac-min-loss}. The point cloud $\mathbb{X}_n$ and persistence diagrams are displayed in Figure \ref{fig:overlap-annulus}. While the VR diagram is noisy, the DTM and KDE diagrams suppress low-persistence features, though at the cost of reduced persistence.

\begin{figure}[t!]
     \centering
     \begin{subfigure}[b]{0.24\textwidth}
         \centering
         \includegraphics[width=\textwidth]{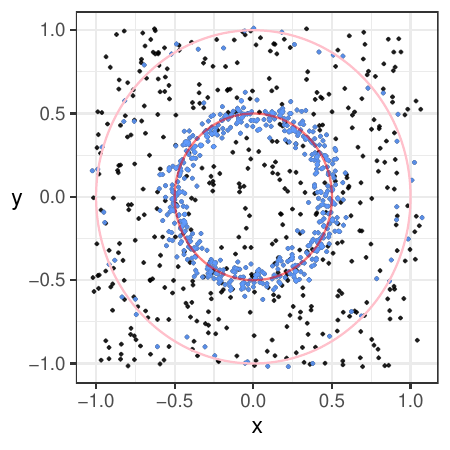}
         \caption{Point cloud $\mathbb{X}_n$}
         \label{fig:overlap_annulus_noisy}
     \end{subfigure}
     \hfill
     \begin{subfigure}[b]{0.24\textwidth}
         \centering
         \includegraphics[width=\textwidth]{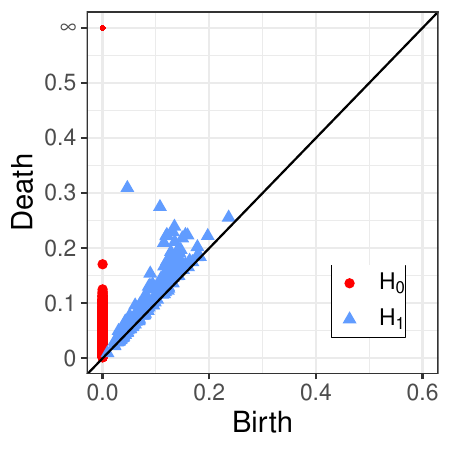}
         \caption{VR diagram}
         \label{fig:vr_pd_overlap_annulus_noisy}
     \end{subfigure}
     \hfill
     \begin{subfigure}[b]{0.24\textwidth}
         \centering
         \includegraphics[width=\textwidth]{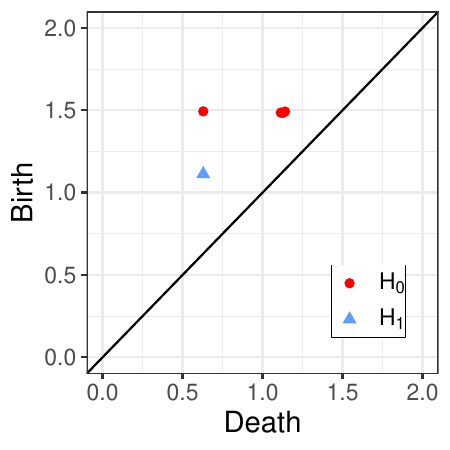}
         \caption{DTM diagram}
         \label{fig:dtm_pd_overlap_annulus_noisy}
     \end{subfigure}
     \hfill
     \begin{subfigure}[b]{0.24\textwidth}
         \centering
         \includegraphics[width=\textwidth]{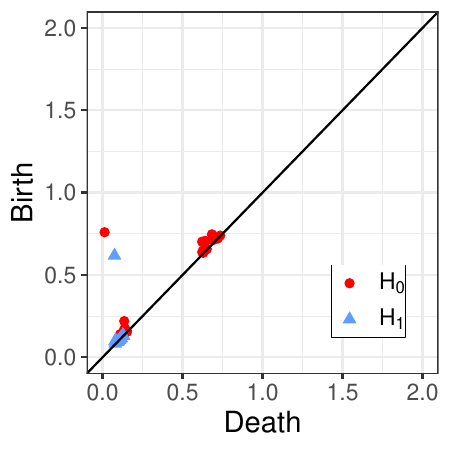}
         \caption{KDE diagram}
         \label{fig:kde_pd_overlap_annulus_noisy}
     \end{subfigure}
    \caption{An illustration of the VR (b), DTM (c), and KDE (d) filtration on the point cloud $\mathbb{X}_n$ (a) (the blue points are signal and the black points are noise).  All three methods identified one dominant $H_1$ feature in terms of persistence. }
    \label{fig:overlap-annulus}
\end{figure}

To demonstrate the topological recovery, $\mathbb{X}^*_n$ was constructed using Algorithm~\ref{alg:subsampling} with a threshold $\lambda$ selected from the range $[0.1, 1]$ and a KDE bandwidth set to the average $k$-NN distance ($k \in [1, 10]$) of points in $\mathbb{X}_n$. These parameters were chosen using the procedure in Section~\ref{sec:parameter-selection} to maximize the most persistent $H_1$ feature for each filtration scheme. Specifically, the optimal $(\lambda, k)$ are $(0.7, 10)$, $(0.4, 2)$, and $(0.6, 8)$ for VR, DTM, and KDE filtration, respectively.  The KDE sample space is shown in Figure~\ref{fig:overlap_annulus_sample_space}, along with the VR (\ref{fig:vr_pd_overlap_annulus_smoothed}), DTM (\ref{fig:dtm_pd_overlap_annulus_smoothed}), and KDE (\ref{fig:kde_pd_overlap_annulus_smoothed}) persistence diagrams computed for $\mathbb{X}_n^*$.
\begin{figure}
     \centering
     \begin{subfigure}[b]{0.24\textwidth}
         \centering
         \includegraphics[width=\textwidth]{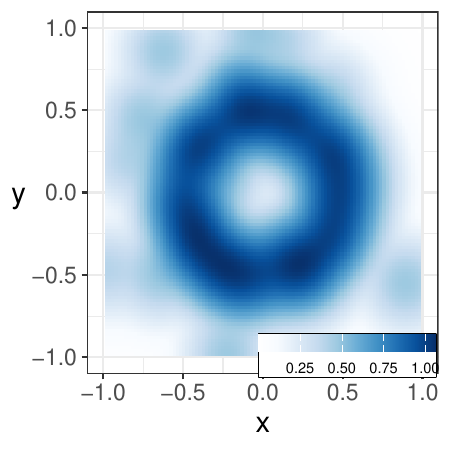}
         \caption{Sample space of $\mathbb{X}_n^*$}
         \label{fig:overlap_annulus_sample_space}
     \end{subfigure}
     \hfill
     \begin{subfigure}[b]{0.24\textwidth}
         \centering
         \includegraphics[width=\textwidth]{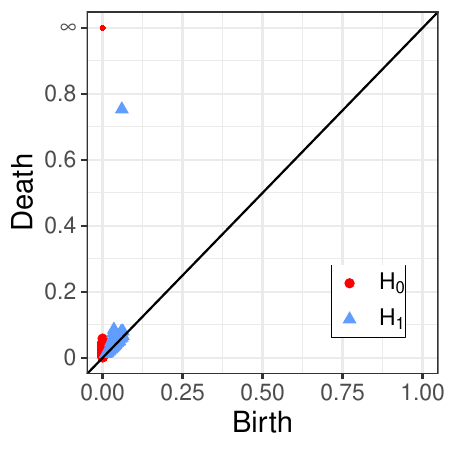}
         \caption{VR diagram}
         \label{fig:vr_pd_overlap_annulus_smoothed}
     \end{subfigure}
     \hfill
     \begin{subfigure}[b]{0.24\textwidth}
         \centering
         \includegraphics[width=\textwidth]{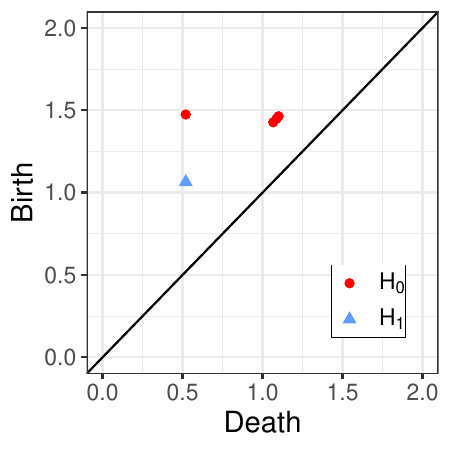}
         \caption{DTM diagram}
         \label{fig:dtm_pd_overlap_annulus_smoothed}
     \end{subfigure}
     \hfill
     \begin{subfigure}[b]{0.24\textwidth}
         \centering
         \includegraphics[width=\textwidth]{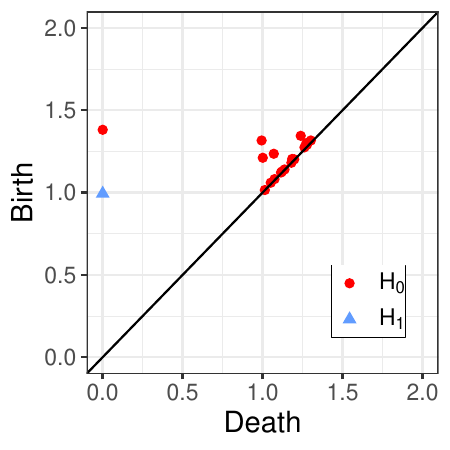}
         \caption{KDE diagram}
         \label{fig:kde_pd_overlap_annulus_smoothed}
     \end{subfigure}
    \caption{An illustration of the VR (b), DTM (c), and KDE (d) filtration on the point cloud $\mathbb{X}_n^*$ from Algorithm~\ref{alg:subsampling}. All three methods identified one dominant and enhanced $H_1$ feature.}
    \label{fig:overlap-annulus-smoothed}
\end{figure}
All three methods revealed one dominant $H_1$ feature. The diagrams from $\mathbb{X}_n^*$ contains fewer low-persistence $H_1$ features than those from $\mathbb{X}_n$.

\subsubsection{Reducing persistence loss in prominent features} \label{subsec:pac-min-loss}

Next, we demonstrate that Algorithm~\ref{alg:subsampling} is robust to variations in the KDE bandwidth (and in the DTM smoothing parameter), that is, once appropriate parameters are selected for Algorithm~\ref{alg:subsampling}, these choices remain effective across different values of $\sigma$ or $m$. Let $\mathbb{X}_n$ be the noisy sample in Figure~\ref{fig:overlap_annulus_noisy}, and $\mathbb{X}_{n,\lambda}^*$ and $\mathbb{X}_{n,0}^*$ be the thresholded and non-thresholded versions, respectively. Let $\mathbb{X}$ be the noise-free data, depicted as the blue points in Figure~\ref{fig:overlap_annulus_noisy}. Figure~\ref{fig:overlap_annulus_max_pers} shows that for a single sample, it is possible to appropriately choose the parameters of Algorithm~\ref{alg:subsampling} (in this case, $\lambda = 0.6$ and the KDE bandwidth is the average $8$-NN distance) such that the maximal persistence associated with $\mathbb{X}_{n, \lambda}^*$ closely approximates the ground truth ($\mathbb{X}$) maximal persistence.

\begin{figure}[ht!]
     \centering
     \begin{subfigure}[b]{0.49\textwidth}
         \centering
         \includegraphics[width=\textwidth]{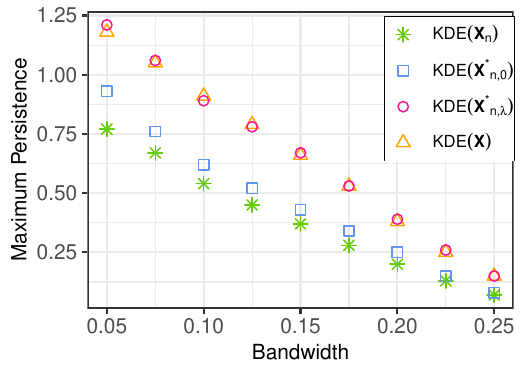}
         \caption{Maximum persistence for a single sample.}
         \label{fig:overlap_annulus_max_pers}
     \end{subfigure}
     \hfill
     \begin{subfigure}[b]{0.49\textwidth}
         \centering
         \includegraphics[width=\textwidth]{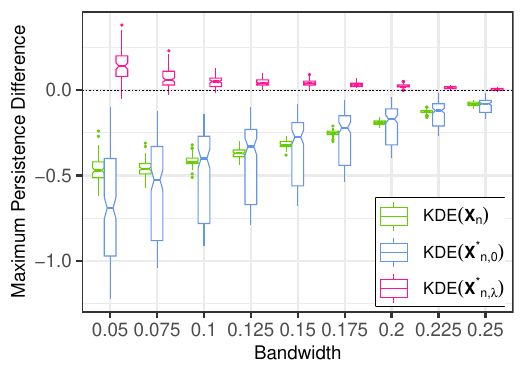}
         \caption{Distribution of the maximum persistence.}
         \label{fig:overlap_annulus_max_pers_diff}
     \end{subfigure}
        \caption{
        MaxTDA estimation results. (a) For an appropriately chosen threshold, the maximal persistence associated with the MaxTDA $\mathbb{X}_{n, \lambda}^*$ (red circles) closely approximates the ground truth ($\mathbb{X}$) maximal persistence (orange triangles). (b) The distribution of the difference in maximal persistence between the three data samples and the ground truth across 100 independent trials, demonstrating that $\mathbb{X}_{n, \lambda}^*$ (red) maximal persistence is less biased. 
        } 
        \label{fig:overlap_annulus_max_pers_}
\end{figure}
The process is repeated $100$ independent times to assess the variability of the construction. The results are presented in Figure~\ref{fig:overlap_annulus_max_pers_diff} as boxplots of the differences between the true and estimated maximal persistence, which indicate that the distributions of the MaxTDA $\mathbb{X}_{n,\lambda}^*$ maximal persistence values are closer to the true values than those from other data spaces. 

\subsection{Data with varying sampling distributions} \label{subsec:sampling-variability}

This section demonstrates how MaxTDA applies to data from topological spaces with similar geometries but different sampling distributions, a scenario that arises, for example, in signal processing when signals at different frequencies are embedded in the same space. Figure~\ref{fig:sampling_variability_point_cloud} shows a $3$D point cloud $\mathbb{X}_n$ with four ellipses of varying density; the goal is to recover the denser ellipse as the ground truth by isolating a single maximally persistent $H_1$ feature. Using the parameter selection procedure in Section~\ref{sec:parameter-selection}, a density threshold of $\lambda=12.22$ was obtained, and a KDE bandwidth was determined as the average 1-NN distance of points in $\mathbb{X}_n$. $\mathbb{X}_{n, \lambda}^*$ and ${\mathbb{X}}_{n, 0}^*$ were constructed using Algorithm~\ref{alg:subsampling}. Figure~\ref{fig:threshold_relxation_max_pers_diff} shows that for bandwidths up to $0.025$, the maximal persistence of KDE($\mathbb{X}_{n, \lambda}^*$) exceeds that of KDE$({\mathbb{X}}_{n, 0}^*)$ and KDE$(\mathbb{X}_n)$. These bandwidths correspond to when the dense ellipse's persistence is at its maximum and increasing, whereas for larger bandwidths the persistence decreases due to over-smoothing, indicating undesirable bandwidths.
\begin{figure}[ht!]
     \centering
     \begin{subfigure}[b]{0.4\textwidth}
         \centering
         \includegraphics[width=\textwidth]{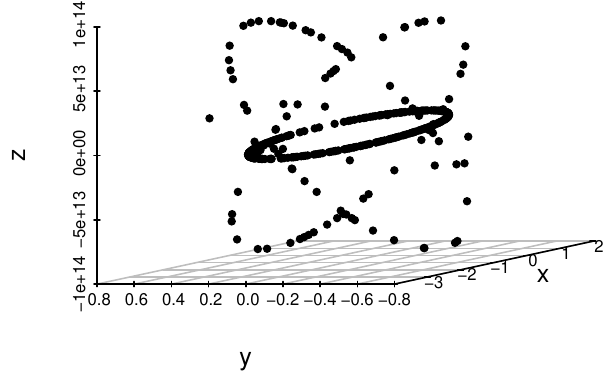}
         \caption{Data $\mathbb{X}_n$}
         \label{fig:sampling_variability_point_cloud}
     \end{subfigure}
     \hfill
     \begin{subfigure}[b]{0.59\textwidth}
         \centering
         \includegraphics[width=\textwidth]{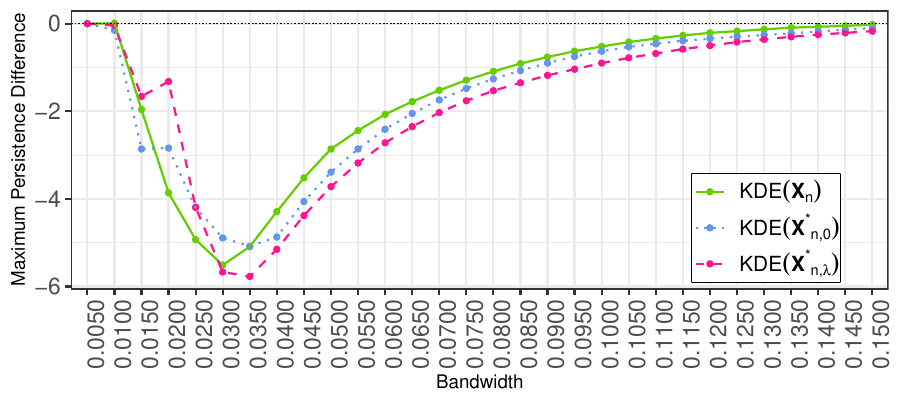}
         \caption{Differences in maximal persistence by bandwidth}
         \label{fig:threshold_relxation_max_pers_diff}
     \end{subfigure}
    \caption{Performance of MaxTDA in estimating the maximal persistence using the sample $\mathbb{X}_{n, \lambda}^*$ compared to the original data $\mathbb{X}_n$ and the non-thresholded sample ${\mathbb{X}}_{n, 0}^*$. (a) Data $\mathbb{X}_n$ with $n = 333$; (b) the difference in the maximal persistence from that of the dense ellipse by KDE bandwidth. } 
    \label{fig:ksampling_variability_threshold_relaxation}
\end{figure}

The optimal bandwidths that maximize the maximal persistence were determined to be $0.02$ for $\text{KDE}(\mathbb{X}_{n,\lambda}^*)$ (denoted hereafter as $\text{KDE}(\mathbb{X}_n^*)$) and $0.015$ for $\text{KDE}(\mathbb{X}_n)$, and these values were used to construct the final persistence diagrams, where the most persistent $H_1$ features in $\mathbb{X}_n$ and $\mathbb{X}_n^*$ had persistences of $1.45$ and $3.18$, respectively. To construct rejection bands, let $\widehat{\text{Dgm}}(\text{KDE}(\mathbb{X}_n))$ and $\widehat{\text{Dgm}}(\text{KDE}(\mathbb{X}_n^*))$ denote the respective persistence diagrams of the upper-level KDE filtrations of $\mathbb{X}_n$ and $\mathbb{X}_n^*$. We bootstrapped $\mathbb{X}_n$ $1000$ times and, for each bootstrap sample $\mathbb{X}_n^{(b)}$, estimated a KDE with $\sigma=0.015$ and computed the bottleneck distance $\hat{t}^{(b)}_{0.015}$ between the $H_1$ features of $\widehat{\text{Dgm}}(\text{KDE}(\mathbb{X}_n))$ and $\widehat{\text{Dgm}}(\text{KDE}(\mathbb{X}_n^{(b)}))$. We also computed the sample $\mathbb{X}_n^{*(b)}$ using Algorithm~\ref{alg:subsampling} at $\lambda=12.22$, with $\mathbb{X}_n^{(b)}$ as the underlying input data, estimated a KDE with $\sigma=0.02$ for $\mathbb{X}_n^{*(b)}$, and computed the bottleneck distance $\hat{t}^{(b)}_{0.02}$ between the $H_1$ features of $\widehat{\text{Dgm}}(\text{KDE}(\mathbb{X}_n^*))$ and $\widehat{\text{Dgm}}(\text{KDE}(\mathbb{X}_n^{*(b)}))$. The $0.95$ quantile of $\{\hat{t}^{(b)}_{0.015}\}$ was $1.3115$ for diagrams from $\mathbb{X}_n^{(b)}$, and that of $\{\hat{t}^{(b)}_{0.02}\}$ was $1.4195$ for diagrams from $\mathbb{X}_n^{*(b)}$, which were used to construct the rejection bands.
\begin{figure}[ht!]
     \centering
     \begin{subfigure}[b]{0.35\textwidth}
         \centering
         \includegraphics[width=\textwidth]{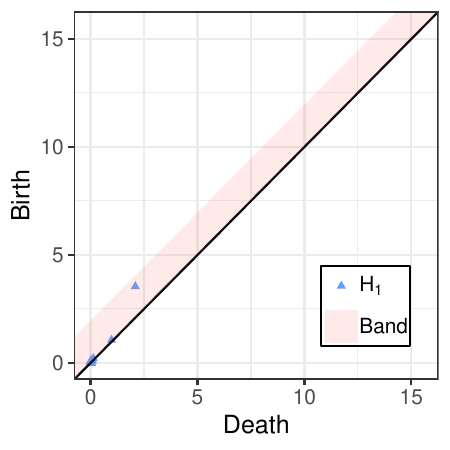}
         \caption{$\mathbb{X}_n^{(b)}$}
         \label{fig:kde_pd_conf_threshold_relaxation_noisy}
     \end{subfigure}
     \begin{subfigure}[b]{0.35\textwidth}
         \centering
         \includegraphics[width=\textwidth]{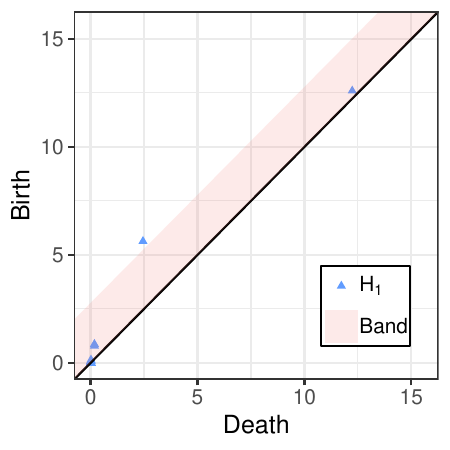}
         \caption{$\mathbb{X}_n^{*(b)}$}
         \label{fig:de_pd_conf_threshold_relaxation_thresholded}
     \end{subfigure}
        \caption{Illustration of the statistical significance of the $H_1$ features based on 1000 bootstrap samples from $\mathbb{X}_n^{(b)}$ (a) and $\mathbb{X}_n^{*(b)}$ (b). The displayed bands (light pink) indicated the 95\% rejection region for the $H_1$ features (blue triangles).  Note that the $H_0$ features have been omitted. } 
        \label{fig:confidence_sets}
\end{figure}
Figure~\ref{fig:kde_pd_conf_threshold_relaxation_noisy} shows the persistence diagram and $95\%$ rejection band for the ordinary KDE, where no $H_1$ feature is statistically significant, while Figure~\ref{fig:de_pd_conf_threshold_relaxation_thresholded} shows the persistence diagram and rejection band for the KDE of the smooth subsamples, in which one statistically significant $H_1$ feature corresponding to the denser elliptical sample is observed. This is partly due to its enhanced persistence, further highlighting the performance of the proposed MaxTDA method.

\section{Exoplanet Data Application} \label{sec:application}
This section explores how MaxTDA enhances periodic time series analysis by linking the persistence of $H_1$ features to signal periodicity. Enhancing the lifetime of $H_1$ features can strengthen a periodicity analysis. We begin by describing a method for constructing a time series representation.

\subsection{Time-delay embedding} \label{subsec:time-delay-embedding}
Time-delay embeddings (TDEs) provide a framework for transforming time series into a multi-dimensional representation \citep{takens2006detecting}. For time series $\{x(t): 0\le t\le n\}$, an embedding matrix is constructed where each row is given by: $\mathbf{v}(t) = \left[ x(t), x(t + \tau), \ldots, x(t + M\tau) \right]$, with time delay $\tau$ and $M+1$ delayed coordinates. Takens' Theorem guarantees that, under suitable conditions, this embedding preserves the shape of the underlying state space if the embedding dimension is sufficiently large \citep{takens2006detecting}. One method for determining $\tau$ is the average mutual information (AMI) \citep{fraser1986independent}. The AMI is computed by partitioning the range of the time series into bins: $\mathcal{I}(\tau) = \sum_{i, j} p_{i, j}(\tau)\log \left( \frac{p_{i, j}(\tau)}{p_ip_j} \right)$, where $p_i$ is the the probability the time series has a value in the $i$-th bin, and $p_j$ is the probability that $x(t+\tau)$ is in bin $j$,  and $p_{i, j}(\tau)$ denotes the probability that $x(t)$ and $x(t +\tau)$ are in the $i$-th and $j$-th bin, respectively. The smallest value of $\tau$ where $\mathcal{I}(\tau)$ reaches a local minimum is chosen as the optimal time delay step. This corresponds to the lag at which the redundancy of information between $x(t)$ and $x(t+\tau)$ is minimized, ensuring that points in the reconstructed embedding space are sufficiently independent. Once $\tau$ is determined, the embedding dimension $M+1$ is selected using Cao's method \citep{cao1997practical}, which evaluates how the structure of the reconstructed space changes as the embedding dimension increases. It identifies the dimension at which the reconstructed space stabilizes. TDEs remain valid under smooth linear transformations, such as principal component analysis (PCA), motivating our subsequent use of PCA for dimensionality reduction \citep{sauer1991embedology}.

\subsection{Exoplanet time series data} \label{sec:exoplanet}
Exoplanets are planets that orbit stars other than our sun.  One method for detecting exoplanets is the radial velocity (RV) method, which measures the forward and backward motion of a possible host star over time. This method was used to discover the first exoplanet orbiting a sun-like star \citep{mayor1995jupiter}. With this RV approach, a certain periodic signature in a star's RV over time suggests the presence of an orbiting exoplanet. The red line in  Figure~\ref{fig:exo-planet-signals} displays a simulated exoplanet RV signal on a circular orbit. Detecting low-mass exoplanets, such as Earth-like planets, remains challenging as their smaller signals can be obscured by stellar activity like star spots \citep{huelamo2008tw, dumusque2016radial, davis2017insights}. 
The green line in Figure~\ref{fig:exo-planet-signals} shows how a simulated star spot using the Spot Oscillation And Planet (SOAP) 2.0 code  \citep{dumusque2014soap} can induce a periodic RV signal that resembles an exoplanet.

While statistical techniques have been developed to detect exoplanets in the presence of stellar variability (e.g., \citealt{rajpaul2015gaussian, dumusque2018measuring, holzer2021hermite, holzer2021stellar, jones2022improving}), they do not fully mitigate the challenges \citep{zhao2022expres}. 
\begin{figure}[ht!]
    \centering
    \includegraphics[width=.6\textwidth,clip=true]{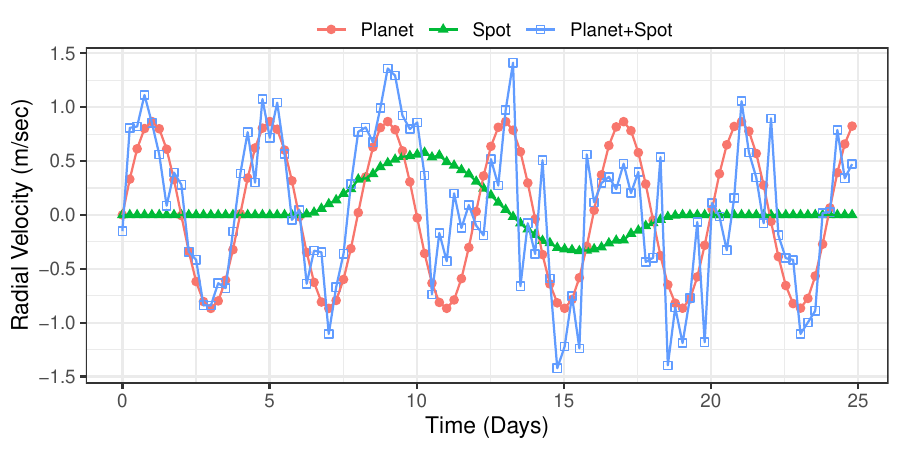}
    \caption{Exoplanet time-series data.  Simulated RV data of an exoplanet (red circles), a 0.05\% spot (green triangles), and the Planet+Spot combined (blue squares).}
    \label{fig:exo-planet-signals}
\end{figure}
This study demonstrates how MaxTDA can help identify and mitigate stellar variability in RV time series analysis; a complete analysis using real exoplanet data is the topic of future research. Our focus is on enhancing feature persistence in combined signals (e.g., Planet+Spot) and assessing the statistical significance of periodic behavior. Using simulated data (Figure~\ref{fig:exo-planet-signals}), we analyze a planet, a star spot, and their combined signal (P+S) RV time series. The spot-induced signal matches the star’s $25.05$-day rotation, while the planet orbits with a $4$-day period and $0.87$ m/sec semi-amplitude. A $0.05\%$ star spot at $30^\circ$ latitude induces a $0.58$ m/sec apparent RV signal. $N(0,1)$ noise was added to ensure the most persistent $H_1$ feature in the combined RV signal is close to the spot's $H_1$ feature before MaxTDA is applied.

TDE matrices were constructed for each time series, with AMI and Cao's used to select $(\tau=4, M=15)$ for the planet and $(\tau=12, M=7)$ for the spot.
\begin{figure}[t!]
    \centering
    \includegraphics[width=1\textwidth,clip=true]{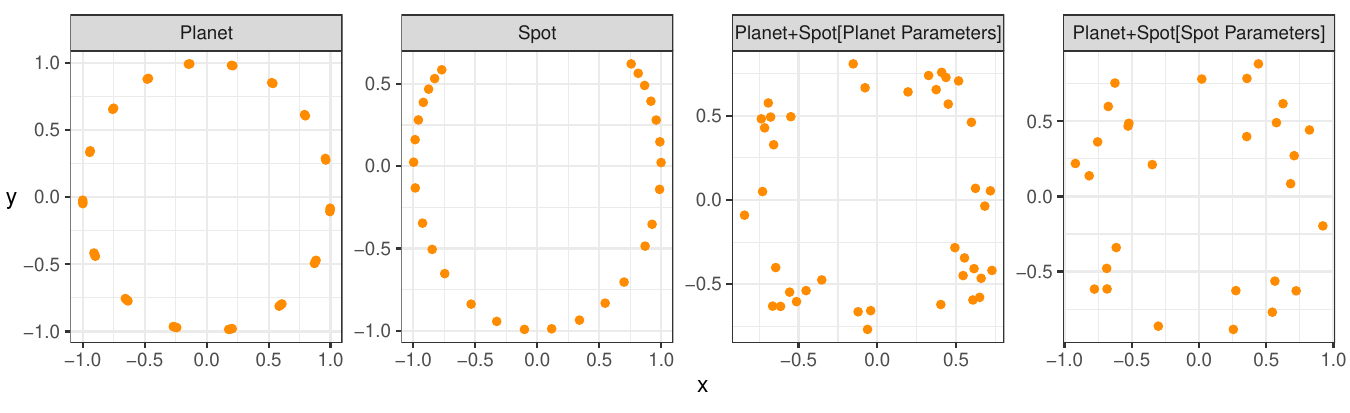}
    \caption{The embedded time series from Figure~\ref{fig:exo-planet-signals}. The Planet and P+S[Planet Parameters] used $\tau=4, M=15$, while the Spot and P+S[Spot Parameters] used $\tau=12, M=6$.}
    \label{fig:four_rv_embeddings}
\end{figure}
Instead of estimating new parameters for the combined signal, we applied the individual embeddings separately, allowing direct comparison of structural and temporal properties. This approach helps assess whether the time series geometry suggests a planet's presence. Each embedding matrix was centered, normalized, and reduced via PCA to two components for analysis (Figure~\ref{fig:four_rv_embeddings}).

\subsection{Quantifying periodicity}
The periodicity of a time series can be assessed using the $H_1$ features of its TDE, where periodic patterns form elliptical shapes in thstate space \citep{perea2015sw1pers}. The roundness of these ellipses, quantified by the maximum persistence of $H_1$ features, serves as a periodicity score: $\max_{(b, d) \in \widehat{Dgm}_1} |d - b|$.  For example, a time series that produces a well-sampled circular loop in its TDE will have high persistence and, therefore, a high periodicity score.

Algorithm~\ref{alg:subsampling} was applied to the P+S[Planet Parameters] and P+S[Spot Parameters] TDEs to reduce noise, with the optimal KDE bandwidth set as the average 1-NN distance. To construct rejection bands, a DTM filtration with $m = 0.01$  was used for the Planet, the P+S[Planet Parameters], the P+S[Spot Parameters], the Smooth P+S[Planet Parameters], the Smooth P+S[Spot Parameters] embeddings, and $m=0.05$ for the Spot embedding, which were selected to maximize the $H_1$ features. Figure~\ref{fig:dgm-signals-planetspot} display the persistence diagrams. The Planet signal has the highest periodicity score ($0.6647$), followed by smoothed P+S[Planet Parameters] ($0.4531$), both statistically significant at the $5\%$ level. The lack of significance in other embeddings is attributed to noise, data distribution variation, and the gap in the Spot's  embedding.
\begin{figure}[ht!]
     \centering
     \begin{subfigure}[b]{0.325\textwidth}
         \centering
         \includegraphics[width=\textwidth]{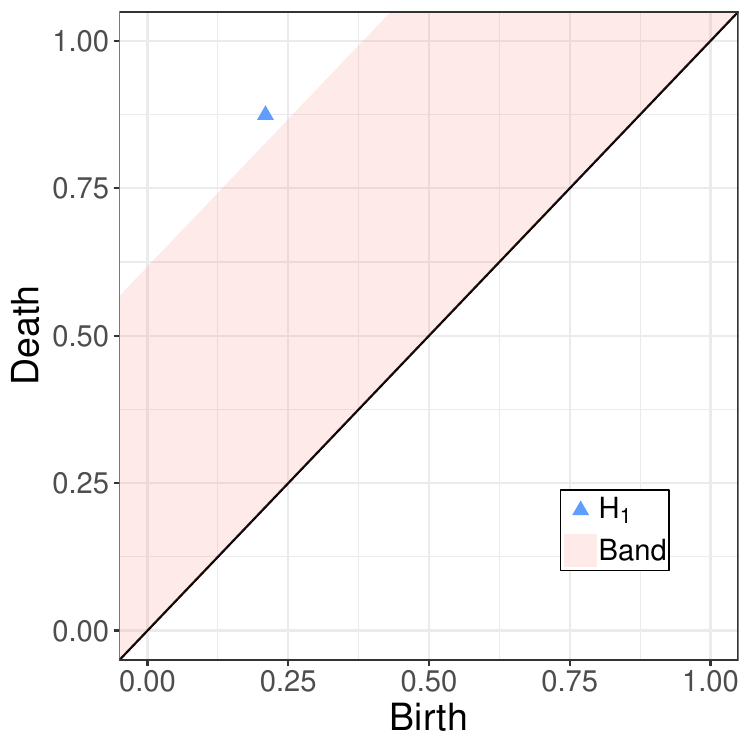}
         \caption{Planet}
         \label{fig:ssp_planet_h1}
     \end{subfigure}
     \begin{subfigure}[b]{0.325\textwidth}
         \centering
         \includegraphics[width=\textwidth]{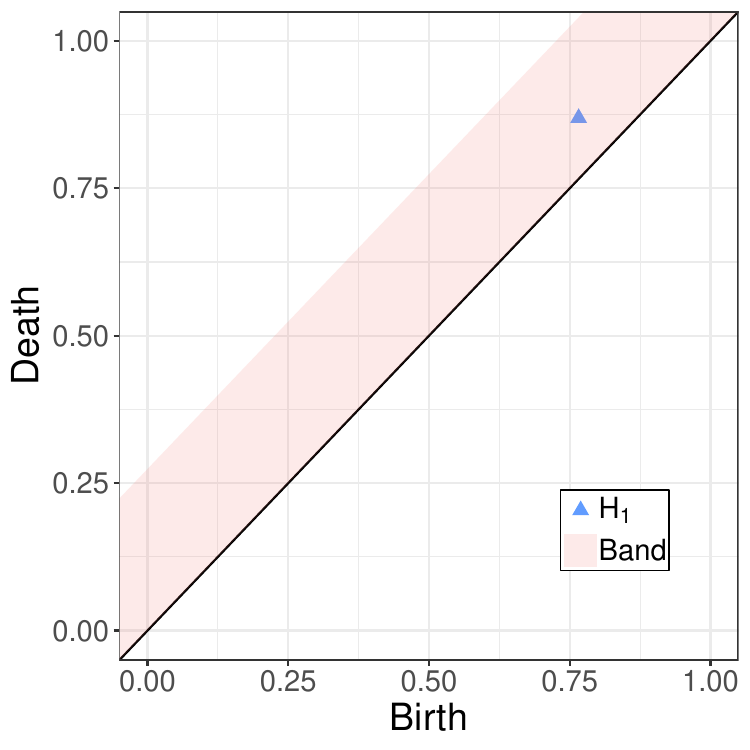}
         \caption{Spot}
         \label{fig:ssp_spot_h1}
     \end{subfigure}
     \begin{subfigure}[b]{0.325\textwidth}
         \centering
         \includegraphics[width=\textwidth]{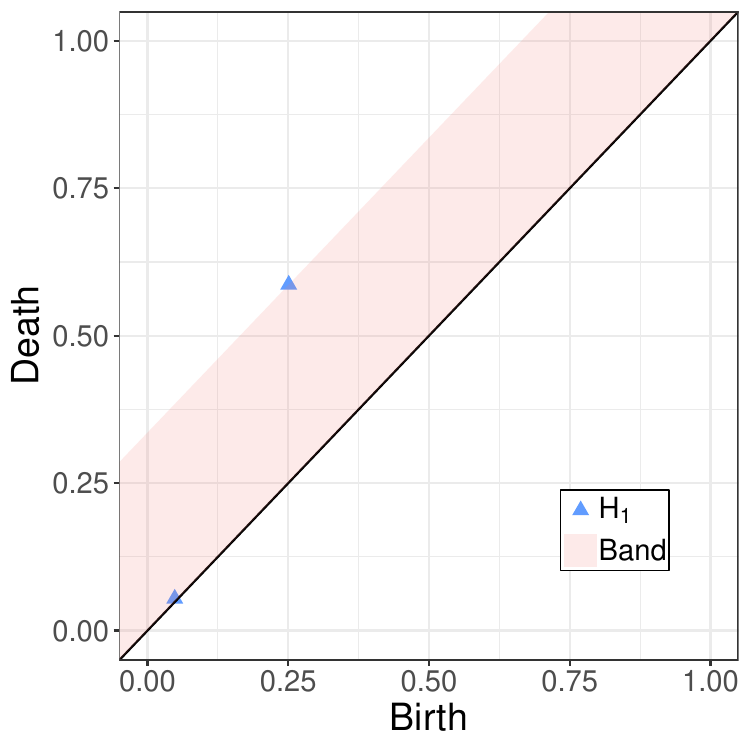}
         \caption{P+S[P]}
         \label{fig:ssp_planetspot1_h1}
     \end{subfigure}
     \hfill
     \begin{subfigure}[b]{0.325\textwidth}
         \centering
         \includegraphics[width=\textwidth]{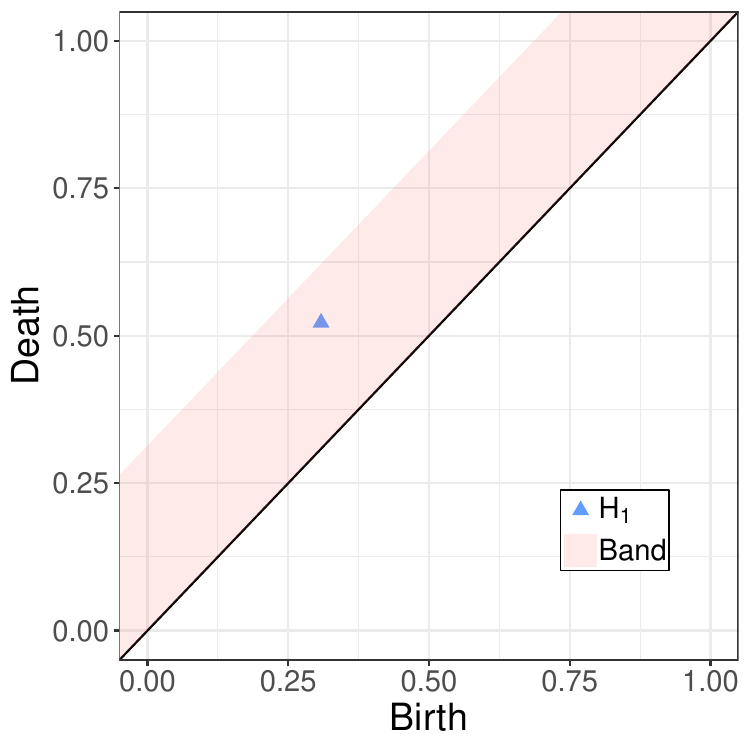}
         \caption{P+S[S]}
         \label{fig:ssp_planetspot2_h1}
     \end{subfigure}
     \begin{subfigure}[b]{0.325\textwidth}
         \centering
         \includegraphics[width=\textwidth]{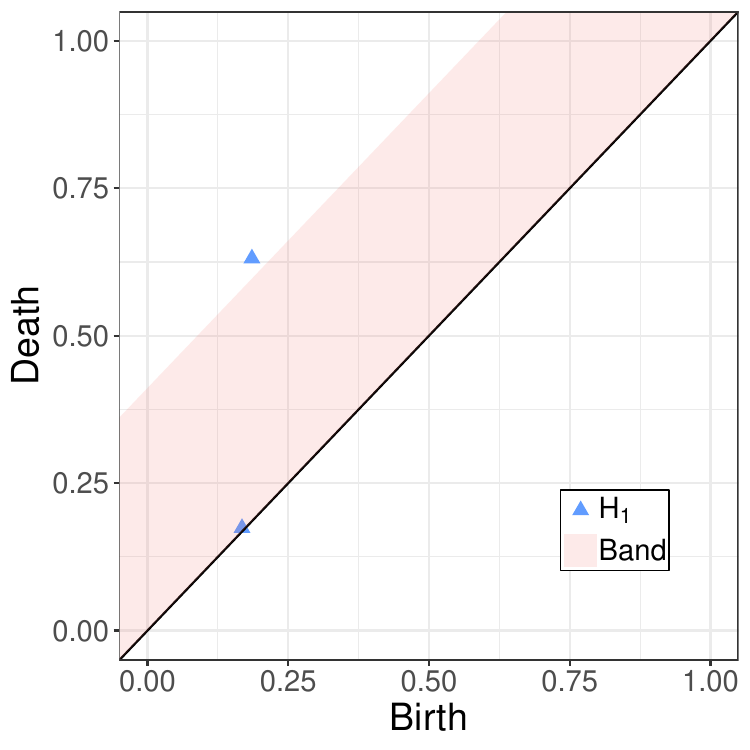}
         \caption{Smoothed P+S[P]}
         \label{fig:ssp_smooth1_h1}
     \end{subfigure}
     \begin{subfigure}[b]{0.325\textwidth}
         \centering
         \includegraphics[width=\textwidth]{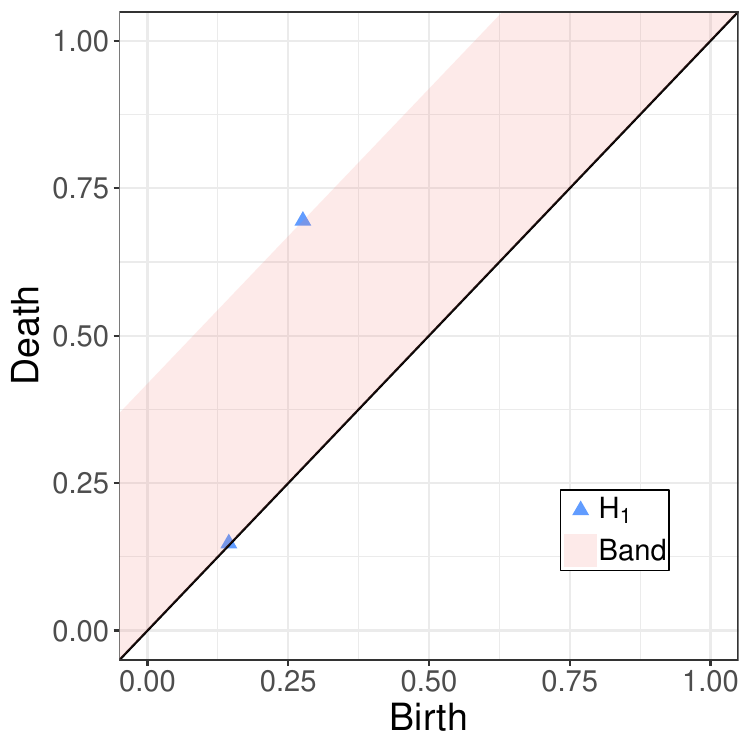}
         \caption{Smoothed P+S[S]}
         \label{fig:ssp_smooth2_h1}
     \end{subfigure}
    \caption{Persistence diagrams (only $H_1$ features) for the Planet (a), Spot (b), and the combined Planet+Spot embeddings and their smoothed versions (c-f) with $95\%$ rejection bands..}
    \label{fig:dgm-signals-planetspot}
\end{figure}
In summary, MaxTDA enhances the $H_1$ feature persistence in the Planet+Spot embedding. This approach is particularly useful for analyzing time series signals with missing observations \citep{dakurah2024subsequence}, embeddings with varying sampling density, or noisy time series where distinguishing or removing noise is impractical or undesirable in the time domain.


\section{Discussion and Conclusion}
\label{sec:conclusion}
This work introduces the MaxTDA methodology that combines kernel smoothing and level-set estimation via rejection sampling to facilitate robust statistical inference for the maximal persistence features in a topological space. The code for implementing the MaxTDA method and replicating the results in this paper is publicly available at \href{https://github.com/Mcpaeis/MaxTDA}{https://github.com/mcpaeis/MaxTDA}. Thresholding the KDEs at a suitable level creates a smooth and dense sampling surface. Rejection sampling is then used to obtain samples that result in improved robustness of estimated homology features with limited reduction in the lifetimes for the maximally persistent feature(s). The maximal persistence estimator is shown to be consistent, and achieves a reduction in bias relative to existing robust TDA methods. The statistical significance of the maximal persistence estimator was assessed via the construction of confidence sets. Several numerical experiments were conducted to illustrate the effectiveness of MaxTDA in uncovering, validating, and drawing meaningful statistical inference for the maximal persistence features of datasets.

There are several important directions for future work and potential improvements. The proposed rejection sampling technique, while effective in low-dimensional settings and geometries that are relatively well-behaved, may face difficulties in ensuring that the sampled points adequately cover the features of interest in complex and high-dimensional data spaces. Dimension reduction methods, such as PCA (e.g., Section~\ref{subsec:time-delay-embedding}) or manifold learning techniques, could be applied as a preprocessing step to improve the effectiveness of the sampling scheme. Alternatively, more adaptive or data-driven sampling strategies could be explored, for example, using importance sampling or Markov chain Monte Carlo approaches that target the most relevant regions of the data space. These adjustments may lead to improved coverage of salient features in higher dimensions, and better stability and efficiency in empirical implementations.

In the exoplanet application in Section~\ref{sec:exoplanet}, a method is proposed to study the contributions of the  planetary signal to the combined signal that includes stellar variability due to a spot.  While this illustration highlights the scientific challenge of detecting low-mass exoplanets in the presence of stellar activity, real RV data can include multiple planets, multiple time-evolving spots, highly irregular time sampling, instrumental effects, and other complexities.
The proposed approach should be considered a preliminary proof of concept requiring further validation across diverse signal scenarios, and serves as an interesting area of future research on topological signal decomposition.

\appendix

\bibliographystyle{apalike} 
\bibliography{references}
\end{document}